\documentclass[%
 reprint,
groupedaddress,
nofootinbib,
 amsmath,amssymb,
 aps,
prb,
]{revtex4-2}

\usepackage{tikz}

\usepackage{babel}
\usepackage{calc}
\usepackage{amsbsy}
\usepackage{amstext}
\usepackage{amsthm}
\usepackage{amssymb}
\usepackage{graphicx}
\usepackage{braket}
\usepackage[normalem]{ulem}

\usepackage[colorlinks=true,
    linkcolor=blue,
    citecolor=red,
    filecolor=magenta,      
    urlcolor=cyan,
    pdftitle={Markov},]{hyperref}
\makeatletter
\theoremstyle{plain}
\newtheorem{thm}{\protect\theoremname}
\theoremstyle{definition}
\newtheorem{defn}[thm]{\protect\definitionname}
\theoremstyle{plain}
\newtheorem*{thm*}{\protect\theoremname}
\theoremstyle{plain}
\newtheorem{conjecture}[thm]{\protect\conjecturename}
\theoremstyle{remark}

\theoremstyle{plain}

\usepackage[T1]{fontenc}
\usepackage[latin9]{inputenc}
\makeatother

\providecommand{\conjecturename}{Conjecture}
\providecommand{\definitionname}{Definition}
\providecommand{\lemmaname}{Lemma}
\providecommand{\remarkname}{Remark}
\providecommand{\theoremname}{Theorem}

\begin{document}
\title{Defining stable phases of open quantum systems}

\author{Tibor Rakovszky}
\affiliation{Department of Physics, Stanford University, Stanford, CA 94305, USA}

\author{Sarang Gopalakrishnan}
\affiliation{Department of Electrical and Computer Engineering, Princeton University, Princeton, NJ 08544}

\author{Curt von Keyserlingk}
\affiliation{Department of Physics, King's College London, Strand WC2R 2LS, UK}

\begin{abstract}
The steady states of dynamical processes can exhibit stable nontrivial phases, which can also serve as fault-tolerant classical or quantum memories. For Markovian quantum (classical) dynamics, these steady states are extremal eigenvectors of the non-Hermitian operators that generate the dynamics, i.e., quantum channels (Markov chains). 
However, since these operators are non-Hermitian, their spectra are an unreliable guide to dynamical relaxation timescales or to stability against perturbations. 
We propose an alternative dynamical criterion for a steady state to be in a stable phase, which we name \emph{uniformity}: informally, our criterion amounts to requiring that, under sufficiently small local perturbations of the dynamics, the unperturbed and perturbed steady states are related to one another by a finite-time dissipative evolution. 
We show that this criterion implies many of the properties one would want from any reasonable definition of a phase. We prove that uniformity is satisfied in a canonical classical cellular automaton, and provide numerical evidence that the gap determines the relaxation rate between nearby steady states in the same phase, a situation we conjecture holds generically whenever uniformity is satisfied. We further conjecture some sufficient conditions for a channel to exhibit uniformity and therefore stability. 
\end{abstract}
\maketitle
\tableofcontents

\section{Introduction}
%We are interested in phases, and phases in open systems. We need to 
The classification of equilibrium phases of matter is a central achievement of condensed matter physics. A particularly deep understanding has been achieved over the last few decades in the case of \emph{gapped phases} of (local) Hamiltonians at zero temperature, which includes both conventional phases (such as discrete symmetry breaking) and a plethora of unconventional ones, from topological order~\cite{WenZoo} to fractons~\cite{NandkishoreFractons}. The present work concerns the generalization of these ideas to the context of non-equilibrium open systems i.e., those evolving under a quantum channel. Can such systems host robust phases separated by phase transitions? How can we define and classify such phases? Indeed, this is a subject with a venerable history~\cite{PhysRevB.75.195331, tauber2014critical, PhysRevA.85.043620, PhysRevLett.110.257204, PhysRevA.86.012116, PhysRevB.94.085150, PhysRevX.5.011017, PhysRevLett.116.245701, kawabata2023lieb,diehl2011topology,bardyn2013topology} and many applications to present-day experiments in ultracold atoms~\cite{RevModPhys.85.553}, polariton condensates~\cite{keeling2011exciton} and digital quantum simulators \cite{midcircuitFeig,midcircuitIBM}. The question is especially timely in the latter context, where recently developed mid-circuit measurement techniques could be used to engineer local channels with exotic steady states. However, the kind of general picture that is available for ground state quantum phases is missing---away from the completely nonlocal limit of open systems described by random-matrix theory~\cite{PhysRevLett.81.3367, collins2010random, PhysRevResearch.4.L022068, denisov2019universal, PhysRevLett.123.234103, can2019random, sa2020spectral, PhysRevX.10.021019, PhysRevB.102.134310, PhysRevX.13.031019, PhysRevB.108.075110}. In this paper we provide a step in this direction by formulating a definition that aims to generalize the notion of ``gapped phases'' to the non-equilibrium context, in such a way that the useful features associated with phases---such as quasi-adiabatic continuity~\cite{hastingswen}---carry over naturally. 

Another reason to consider the stability of phases in open systems, related to the aforementioned mid-circuit measurements, is their close connection to self-correcting (classical or quantum) codes~\cite{Dennis2002,harrington2004analysis,pastawski2011quantum}. If the dynamics has multiple steady states, these can be used to encode information (whether classical or quantum will depend on the steady state manifold) and stability to perturbations readily translates into the robustness of such an encoding against various types of noise~\cite{liu2023dissipative}. In such self-correcting codes,  the error-correction is achieved through local dissipation. In contrast, active error correction involves the application of non-local channels, where typically the non-locality comes from a  background layer of classical communication and processing.  As quantum computers increase in size, this background classical communication and processing might become a bottleneck \cite{breuckmann2017local}, wasting precious time in a situation with limited coherence time. Local self-correcting codes might be a useful alternative in these regimes. Thus, understanding the criteria for stable non-equilibrium phases is of considerable practical importance as well.
\begin{figure}
    \centering
    \includegraphics[width=1.01\columnwidth]{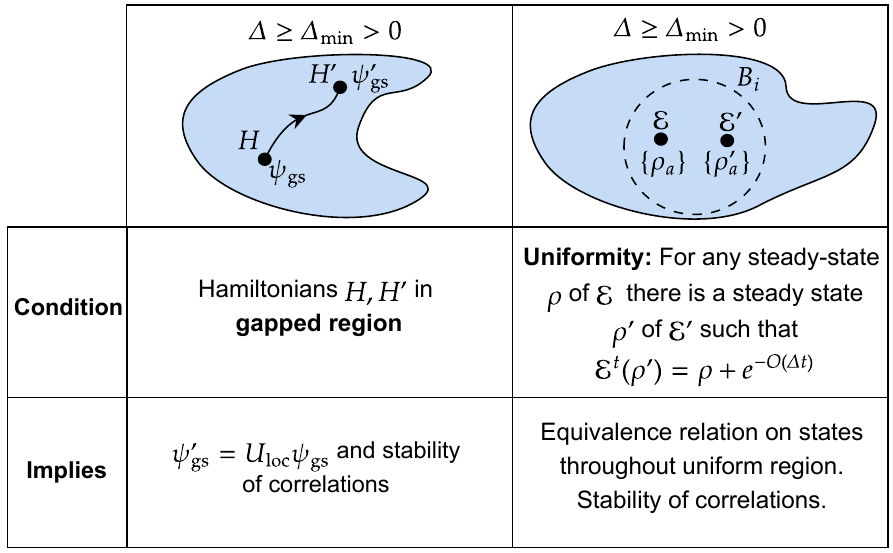}
    \caption{Summary of our definition of ``uniformity'' for quantum channels, in comparison to the notion of gapped phases in equilibrium.}
    \label{fig:uniformity}
\end{figure}

In this work, we address this issue by formulating a condition on families of quantum channels $\mathcal{E}$ which is sufficient for them to exhibit many of the features we expect from a reasonable definition of a phase. Stated colloquially, this condition, which we call \emph{uniformity}, says that the steady states of $\mathcal{E}$ relax exponentially quickly to the steady states of a perturbed channel $\mathcal{E}'$ when evolved with the latter. We show that, when combined with Lieb-Robinson bounds~\cite{poulin,nachtergaele2011lieb}, uniformity implies many of the features we have come to associate to stable phases of matter, such as the smoothness of local observables and the robustness of long-range correlations. To further justify our definition, we prove that it is satisfied by at least one example of a non-trivial phase in one dimension. 

In ground state phases, the spectral gap often plays an important role in proofs of stability. As we discuss below, for open systems the situation is more complicated and their stability to perturbations depends also on the structure of their eigenstates, an issue that does not arise in the hermitian case. Nevertheless, we provide numerical evidence that in the aforementioned one-dimensional example, the rate of exponential relaxation to the new steady state is in fact set by the spectral gap. We also discuss implications of uniformity on the spectrum more generally, relating it to a condition on the spectral resolvent of the quantum channel in question.

A weakness of the uniformity condition is that it refers to an entire region in the parameter space of local channels, rather than establishing the stability of a single channel $\mathcal{E}$ to perturbations. In the last part of the paper, we address this weakness, and conjecture sufficient conditions on a channel $\mathcal{E}$ such that it exhibits phase stability. Based on the study of some concrete examples in one dimension, we distinguish between \emph{perturbative} and \emph{non-perturbative} instabilities and argue that the former are absent whenever the channel $\mathcal{E}$ is able to erode errors that occur on top of a steady state. We further conjecture that absolute stability obtains whenever this erosion process is sufficiently fast, and $\mathcal{E}$ only has a finite number of distinct steady states.  

In summary, our work clarifies a structure sufficient for open systems to form robust phases of matter and establishes a framework that aims to put the theory of such phases on a footing equal to their much studied equilibrium counterpart. We expect that this framework will prove useful in uncovering novel quantum phases that are particular to non-equilibrium open quantum and classical systems and might be relevant to various experimental settings. 

The rest of this work is organized as follows. In Sec.~\ref{sec:summary} we give a more detailed outline of our motivation and results, also situating them in the context of previous literature. In Sec.~\ref{background} we review the notions of quasi-adiabaticity for ground states and thermal equilibrium states, and introduce a simple illustrative example of an open system with a spectral gap but a diverging relaxation time. In Sec.~\ref{uniformity} we introduce uniformity as our local criterion for a stable phase, and derive some of its consequences. In Sec.~\ref{sec:stavskaya} we rigorously establish uniformity for a canonical example of a stable classical cellular automaton, and numerically explore the relation between the local relaxation time and the gap. The argument we present seems quite general, and seems to extend to other perturbed cellular automata. In Sec.~\ref{sec:erosion} we discuss a physical mechanism for stability---based on the erosion of perturbations---and conjecture a criterion for stability that can be expressed purely in terms of properties of the unperturbed dynamics. Finally in Sec.~\ref{discussion} we conclude with a summary and a discussion of how our results apply to stable phases that can perform passive quantum error correction. 

\section{Connections to previous literature}\label{sec:summary}

In this section, we provide more details on the motivation for our present work, and its relation to previous literature  defining phases of matter in and out of equilibrium. 

Before delving into the non-equilibrium case, let us briefly review the standard concept of gapped ground-state phases.
Two quantum states $\ket{\psi_1}, \ket{\psi_2}$ are said to be in the same phase if there exists a finite-depth local unitary (FDLU) circuit $U$, consisting of finitely many layers of geometrically local unitary operators, such that $U \ket{\psi_1} \approx \ket{\psi_2}$\footnote{Here $\approx$ is meant to imply that the expectation values of observables with finite support can be made arbitrarily small, even though the overlap $\braket{\psi_2|U|\psi_1}$ is still generally zero in the thermodynamic limit.}\cite{Chen10}. Being in the same phase is an equivalence relation, since $U^\dagger$ is a finite-depth unitary circuit if $U$ is. By construction, $U$ only spreads correlations over a finite distance, so it cannot change the long-distance asymptotics of correlations, entanglement, etc. Therefore these long-distance properties are robust signatures that all wavefunctions in a phase share. Alternatively, one can define a \emph{gapped} phase in terms of Hamiltonians: two Hamiltonians are in the same gapped phase if connected by a path through parameter space along which the spectral gap remains open. These two definitions are related to each other by the adiabatic and Lieb-Robinson theorems, which taken together imply that when two Hamiltonians are connected by a gapped path, their ground states lie in the same phase~\cite{hastingswen,Hastings_Locality}. Note, however, that the implication only goes one way: two Hamiltonians whose ground states are in the same phase (by the FDLU definition) might not be connected by a gapped path. Indeed, one can even find examples where a gapped and a gapless Hamiltonian share the exact same ground state~\cite{muller1985implications,freedman2005line}.

The upshot is that if the gap remains open in a region of parameter space  (Fig.~\ref{fig:uniformity}), then all the associated ground states are in the same phase in regards to their asymptotic correlations and entanglement structure. Inside such a gapped phase, local observables are analytical functions of the parameters in the Hamiltonian. One might also ask whether a particular gapped Hamiltonian model lies in a stable phase. The reasoning above shows that it is sufficient for the gap to remain open in a region around the Hamiltonian. For perturbations that act only in a finite-size region of the system, this stability can be proved and is known as the principle that local perturbations perturb locally (LPPL)~\cite{michalakis_lppl,deRoeck2015local}. For perturbations that act everywhere, the gap is in fact not always stable, but can be proven to be so under appropriate conditions~\cite{michalakis} and can be explicitly checked in many others. 

Less well understood is the case of equilibrium quantum phases of matter at finite temperature. To our knowledge, no general notion of the unitary quasi-adiabatic maps exists in that case, although some related ideas have been explored in Ref. \cite{Hastings_belief}. We review these in App.~\ref{App:thermal_2D_Ising} where we also prove some new results, providing sufficient conditions for thermal states of slightly perturbed Hamiltonians to have qualitatively similar correlation functions---see also the recent independent work \cite{capel2023decay}  which discusses some similar results. Finally, we mention Ref. \cite{Gibbs_samplers} which provides a recipe for constructing a local Lindbladian that has a particular Gibbs state as its steady state. This brings the question of thermally stable phases into the purview of phases of open quantum systems, which is the question we turn to now.

Is there a notion of a stable or gapped phase for the steady states of open systems (classical or quantum), and what do we even mean by phase in this context? Our intuitive picture of a phase is that it is a contiguous region of parameter space where long-distance properties of the system remain similar. One  context where a notion of stability has been explored is that of classical \emph{probabilistic cellular automata}~\cite{Lebowitz1990,ponselet_thesis}. A famous example is given by Toom's rule~\cite{toom1980stable}, a two-dimensional cellular automaton exhibiting robust bistability: it has two macroscopically distinct steady states which resemble those of equilibrium Ising ordered phases (e.g., 2D Glauber dynamics at low temperature and zero field). However, whereas the bistability of Glauber dynamics is fragile to Ising symmetry breaking perturbations, in the case of Toom these steady states are stable \emph{arbitrary} perturbations, without the need for any symmetry constraint. G\'acs has developed an even more mind-boggling example, which hosts an exponentially large number of stable steady states in one dimension~\cite{Gacs2001}. 

These examples consider deterministic cellular automata perturbed by weak noise, and use properties of the deterministic dynamics to argue that the set of steady states is robust. Away from this limit, for more general local Markov chains---and their quantum versions, i.e., quantum channels---the precise definition of what it means to have a stable phase, and conditions needed to achieve it, is much less clear. While superficially this task resembles the ground state classification problem, as both concern extremal eigenvectors of the operator generating the dynamics, they are made very different by the fact that for open systems, the operator in question is non-Hermitian. Indeed, the obvious extension of the definition of a phase as an equivalence classes of states fails: if two steady-state density matrices $\rho_1$ and $\rho_2$ are related by a finite-depth quantum channel $\Phi(\rho_1) \simeq \rho_2$, the reverse relation is not true in general, since $\Phi^{-1}$ need not be a quantum channel. As an extreme example, one can connect all steady states to the trivial infinite-temperature state by applying a short-depth depolarizing quantum channel to an arbitrary state.

A natural fix to this issue~\cite{Coser_2019,ma2023average}, is to say that density matrices $\rho_1$ and $\rho_2$ are in the same phase if they are related by low-depth quantum channels in both directions. That is, there exist low ($\mathrm{poly}\log(L)$) depth circuits of local channels $\Phi,\Phi'$ such that $\Phi(\rho_1) \simeq  \rho_2$ \emph{and} $\Phi'(\rho_2) \simeq \rho_1$. This has many of the properties one would wish for as a definition of a phase of matter, and implies that the local correlation properties of $\rho_{1,2}$ are equivalent. While this establishes an equivalence relation for states, there remains a question whether there is an appropriate definition for phases of \emph{channels}, analogous to the ``gapped path'' condition of Hamiltonians, from which the equivalence of steady states, and other desirable properties, would follow. This is the question that our paper aims to address. We do so by introducing the notion of \emph{uniformity} in the space of channels, which we propose as an analog of ``gapped phase of matter'' of Hamiltonians.

Uniformity takes its inspiration from Refs. \cite{RapidMixing,Cubitt2015}, which provided a sufficient condition for the stability of Markovian open systems. 
The quantum channels considered there satisfied a \emph{rapid-mixing} property, namely that \emph{any} initial state gets very close to a steady state after a time that scales at most logarithmically with system size (see Refs. \cite{RapidMixing,Cubitt2015} for a more precise definition).  Using Lieb-Robinson bounds (which quantum channels do obey~\cite{poulin}), the authors were able to show that this naturally leads to the stability of various steady state properties, such as the expectation values of local observables. 
However, rapid mixing is clearly at odds with having a nontrivial steady state: starting from a trivial state, the Lieb-Robinson bound does not allow long-range correlations to form on the timescale it takes the system to reach its steady state. Thus, it remains an open problem to find conditions for stability that can encapsulate phases with non-trivial long-range order.  
Uniformity is similar in spirit to rapid mixing, but is designed to avoid the above pitfall. Instead of requiring fast relaxation starting from arbitrary initial states, we require it only for steady states of other channels that are nearby in parameter space. In fact, to get around the issue of non-equivalence mentioned above, we will require this to be true for any pair of channels within some small but finite region of parameter space: i.e., we will require that a steady state of any channel within this region relaxes exponentially quickly when evolved with any of the other channels within the same region. While this is still a strong assumption, we will show that it holds at least one known non-trivial example of a probabilistic cellular automaton, and argue non-rigorously that in fact it applies to a larger family of such models, including Toom's rule. Uniformity is essentially designed to capture a notion of adiabatic continuity and indeed, we show that many of the properties associated to stable zero temperature phases follow from it, including LPPL, the analyticity of local expectation values and the robustness of long-range order. 

%In ground state phases, we can invoke the spectral gap to argue for adiabatic continuity; is there a similar interpretation of uniformity in terms of the spectral properties of the underlying channels? 
Is there a spectral interpretation of uniformity that would make it appear more similar to the gap condition of Hamiltonians? For channels, the gap is directly related to a physical relaxation time:  for a fixed system size, in the late-time limit, the state of the system is $\rho(t) = \rho_0 + \exp(-\Delta t) \rho_1 + \ldots$, where $\Delta$ is the gap, and $\rho_1$ is the corresponding eigenstate. Crucially, the caveat of having to consider the long-time before the thermodynamic limit is very important here. If we could apply this formula directly in the thermodynamic limit, it would imply rapid mixing, which we argued above obtains only in the trivial phase. The possibility of non-trivial gapped phases arises due to the non-commutativity of the two limits. This comes about when the generator of the Markov process has (nearly) parallel eigenvectors~\cite{trefethen2005spectra,PhysRevLett.121.086803, zhang2022review, PhysRevX.11.031019}. This leads to the appearance of additional prefactors when we expand $\rho(t)$ in terms of eigenstates and when these prefactors have a sufficiently strong system-size dependence, they can compensate the gap, leading to slow relaxation. In this case, the gap determines the relaxation time only for states that are already sufficiently close to the steady state, but not for arbitrary initial states. This relates to our notion of uniformity: the steady states of nearby channels are already close and thus rapidly approach some perturbed steady state, while trivial (short-range correlated) initial states are far away and take a long time to relax despite the gap. We provide clear numerical evidence that the gap of the channel governs the relaxation rate of the steady states of nearby channels, i.e., the rapid timescale that enters our uniformity assumption. 

It might seem surprising that stable phases should be associated with channels that are nearly defective matrices, since this seems to imply an extreme sensitivity to perturbations. The resolution must lie in the structure of eigenstates, such that while the spectrum is susceptible when arbitrary perturbations are considered, it is actually stable when we restrict to ones that obey locality (along with potentially some other constraints, such as symmetries). Indeed, we will relate uniformity to a spectral property in terms of a resolvent, which is similar to a spectral gap condition but also explicitly includes the family of allowed perturbations. 

Much like the definition of gapped phases of Hamiltonians,  our notion of uniformity has the weakness that it is a condition on an entire open region of parameter space, rather than to a single channel or Markov process. Ideally, one would like some condition of stability that means that a particular channel is stable to some set of perturbations, similar to the results that exist for some classes of gapped Hamiltonians~\cite{michalakis}. We analyze some known cases of stable and unstable cellular automata to conjecture such a condition. The key idea is that local perturbations applied to a steady state should decay quickly, such that they cannot accumulate and change the character of the state during the dynamics. This is consistent with the arguments made above: states that only differ from a steady state by a local perturbation should be ``close enough'' to it such that the gap appropriately determines the rate of relaxation.  The picture that emerges from our analysis is that, in a non-trivial gapped phase, local perturbations relax fast (on a timescale that roughly tracks the spectral gap), while \emph{macroscopic} perturbations relax on parametrically slower timescales that diverge with system size. It is this combination of spectral properties and locality constraints that leads to steady states that are non-trivial yet stable and we formulate some conjectures for conditions that are sufficient to ensure stability to arbitrary perturbations. 

\section{Background and basic results}\label{background}

This section is organized as follows. We begin by reviewing the concept of local quantum channels, mainly to fix the notation we will use in the rest of this paper. Next, we discuss an illustrative example of a channel with a nontrivial gapped phase: namely, a single particle undergoing a biased random walk. Finally, we show that \emph{any} channel in a nontrivial phase is unstable to specific perturbations.

\subsection{Quantum channels and classical Markov chains}

This work considers the properties of both classical and quantum systems with local Markovian dynamics. The former are a subset of the latter (see below), so we will set up the more general quantum formalism here.
We will typically consider states (i.e., density matrices) on Hilbert spaces consisting of qubits on a lattice. The observables we consider are those with finite support and finite norm. The expectation value of an observable $o$ in a state $\rho$ is denoted $\braket{o|\rho}= \mathrm{tr}(o^\dagger \rho )$. 

The most general Markovian evolutions on such states are so-called completely positive trace preserving (CPTP) maps or quantum channels: $\ket{\rho} \rightarrow \mathcal{E}\ket{\rho}=\ket{\mathcal{E}(\rho)}$. All such maps have the property that they preserve the identity matrix on the left $\bra{I}\mathcal{E}=\bra{I}$; this is equivalent to the requirement that they preserve probability i.e., $\mathrm{tr}(\rho)=1$. A quantum channel furnishes a discrete time evolution $\rho = \mathcal{E}^t (\rho)$. There is also a continuous time version of quantum channels; these are the celebrated ``Lindbladian'' maps $\mathcal{L}$. The distinction between channels and Lindbladians is not important for the present work; our results apply to either.

We consider a subclass of quantum channels, namely those with ``locality'' in the Lieb-Robinson sense \cite{hastings_markov_LR,poulin,nachtergaele2011lieb}. This includes as a subclass unitary evolution with a local Hamiltonian and local circuits consisting of nearest neighbour channels. We will define local channels as being those which obey a Lieb-Robinson bound, which we define in terms of an observable $o_X$ supported on some subset of the lattice $X$, and a CPTP map  $\mathcal{V}_Y$ supported on $Y$. The two subsets are separated by some distance $d_{XY}$ on the lattice and the LR bound can be stated as
\begin{equation}\label{eq:LR_bound}
    \braket{o_X | \mathcal{E}^t (\mathcal{V}_Y-I)|\rho}\leq C V \min(1,e^{\alpha(t - d_{XY}/v)}),
\end{equation}
which should hold for any state $\rho$, any norm $1$ observable $o_X$ and any choice of $\mathcal{V}_Y$. $C,v,\alpha$ are constants that depend only on on lattice details and $\mathcal{E}$ and $V=\min(|X|,|Y|)$ is the minimum volume of sets $X,Y$. Intuitively, the LHS of Eq.~\eqref{eq:LR_bound} measures the effect of perturbing a state $\rho$ locally at position $Y$ and measuring the effect at a later time $t$ at position $X$; the bound means that this effect should be small as long as $vt < d_{XY}$. As \cite{poulin} showed, this property holds for any dynamics generated by a local Lindbladian. 

Finally we note that classical Markov chains can always be realised as CPTP maps. Suppose one has a discrete time Markov chain with transition rule $\ket{b} \rightarrow\sum_{b'} T_{b' b}\ket{b'}$, where each $b,b'$ is a state of the Markov chain (e.g., bit-string configurations on a lattice). We define an associated CPTP map as
\begin{equation}
    \mathcal{E}(\ket{b}\bra{b''})=\delta_{b b''}\sum_{b'} T_{b' b}\ket{b'}\bra{b'}.
\end{equation}
This embeds the Markov dynamics into a quantum channel, where the classical probability distributions are represented as density matrices diagonal in the $\ket{b}$ basis. We will call $T$ a local transition rule if the associated CPTP map is local. Many of the examples considered below will include deterministic cellular automata (CA) as special points; these are a subclass of Markovian processes with  $T_{b' b} = \delta_{b'=f(b)}$ where $f$ is a map on configuration space.

\subsection{A simple nonequilibrium example: the asymmetric exclusion process}\label{sec:ASEP}

\begin{figure}
    \centering
    \includegraphics[width=0.85\columnwidth]{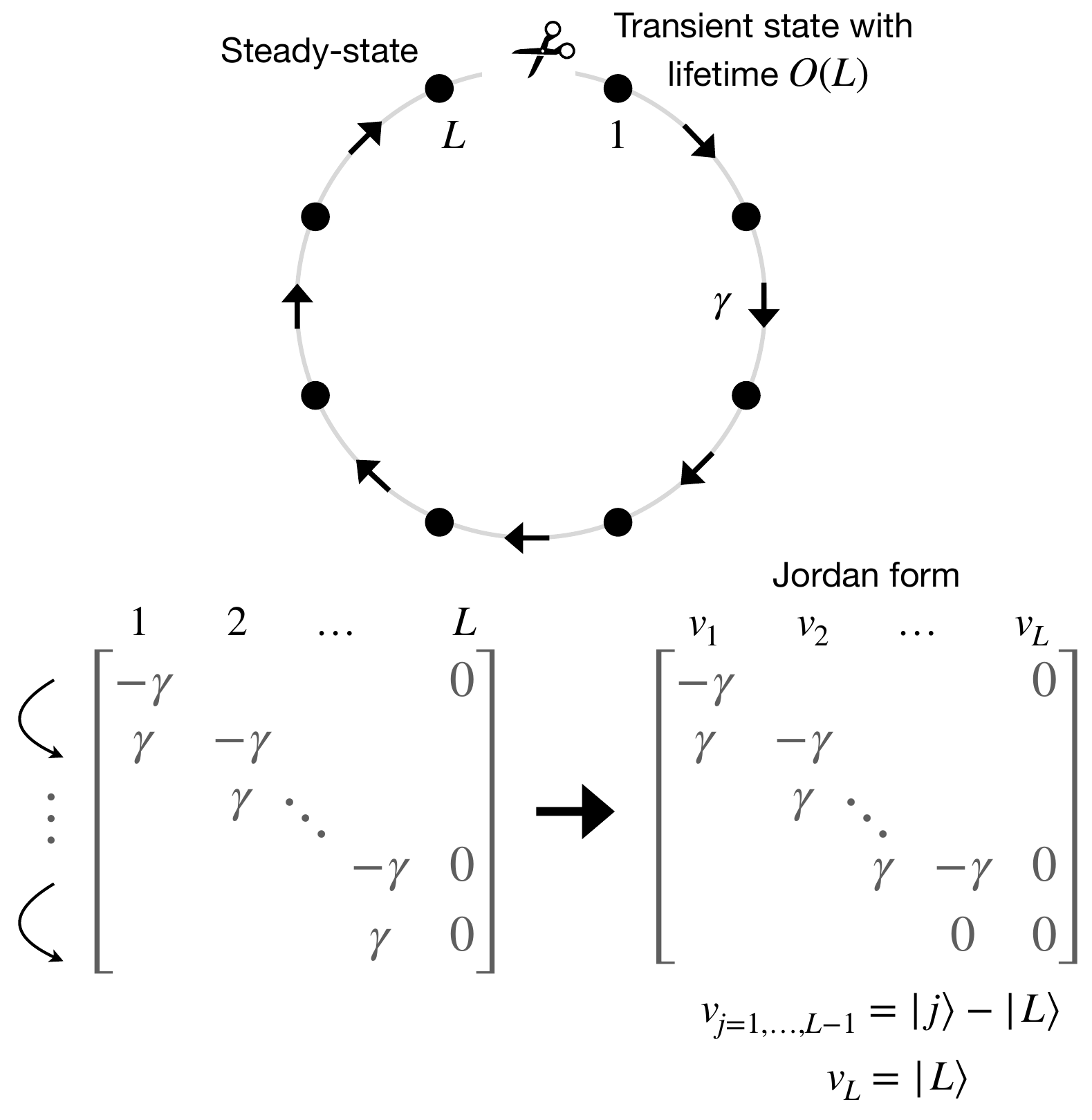}
    \caption{TASEP model and its perturbations. The figure shows the transition matrix on a ring of $L$ sites with the coupling between $x=1,L$ removed. There is a rightward biased hopping rate $\gamma$ on each site, except the last site $x=L$, which is an absorbing steady-state. The resulting matrix (shown) is gapped. However, the system exhibits long relaxation time scales. For example, a particle initialized at $x=1$ takes an $O(L)$ time to relax to the steady state. The gap and some local observables are unstable to introducing an infinitesimal coupling between between the steady state $x=L$ and long lived transient states like $x=1$. }
    \label{fig:tasep_fig}
\end{figure}

We now turn to an illustrative example of a  non-equilibrium process that captures some of the ideas mentioned in our introduction, and will guide our intuition in subsequent sections. In this example, the matrix $M$ generating the dynamics is gapped. However, the gap only sets the relaxation rate of states that are already `close' to the steady state; there are other initial states that take extensively long to relax \cite{Ueda20}. We show that the presence of these long-lived states lead to the gap (and other observables) being unstable to perturbations that couple them to the steady state. This instability is diagnosed by the behavior of the projected spectral resolvent of $M$ (Eq.~\eqref{eq:resolvent}), which diverges despite the presence of a gap. 

The example we consider is a Markov process\footnote{For the purposes of this section we work with a continuous-time dynamics as it makes the analogy to Hamiltonian perturbation theory more explicit; however, our conclusions also apply to discrete-time channels.} involving a single particle on a ring of length $L$. The particle undergoes a biased random walk. However, we remove the coupling between the first and last site in the chain, thus effectively imposing open boundary conditions. This is the one-particle limit of the familiar \emph{asymmetric exclusion process} (ASEP). At each site $i (\neq 1, L)$ in the bulk of the chain, the dynamics is given by the rate equation $\dot{p}_i = \gamma_R p_{i-1} + \gamma_L p_{i+1} - (\gamma_L + \gamma_R) p_i$, i.e., the particle hops to the left (right) at rate $\gamma_L$ ($\gamma_R$). At the first (last) site, the rate for hopping left (right) is set to zero. It is simple to check that this model has two phases in the limit $L \to \infty$: when $\gamma_L > \gamma_R$, the probability distribution in the steady state is localized at the left edge, and vice versa. Moreover, the spectrum of this Markov process is gapped in either phase, with a gap $\Delta=(\sqrt{\gamma_L} - \sqrt{\gamma_R})^2$\cite{Ueda20}. Away from this single-particle limit, the ASEP can be solved using Bethe ansatz methods and retains a gapped phase for open boundary conditions~\cite{PhysRevLett.95.240601}. The mechanism in the many-body case is fundamentally the same as in the simpler one-body case.

Naively, one might think that the presence of a gap implies that any initial state relaxes to the (unique) steady state exponentially fast. This, however, is in contradiction with the locality of the dynamics: for $\gamma_R > \gamma_L$,  light-cone bounds imply that a particle initialized at the first site takes a time at least $O(L)$ to relax to its steady state near the last site. The resolution to this puzzle was understood in \cite{Ueda20}, and is most transparent in the extreme limit where $\gamma_L=0$ (conventionally called the \emph{totally} asymmetric exclusion process, or TASEP). Here, the transition matrix of the Markov process becomes lower triangular and all its nontrivial eigenvectors coalesce into just two eigenvectors with eigenvalues $0,-\gamma_R$ respectively (see Fig.~\ref{fig:tasep_fig}): The steady state, and the eigenvector corresponding to a large Jordan block of size $L-1$.  The initial state with probability localized at the left end ($i=1$) of the chain is orthogonal to the eigenvectors of the matrix, and its early-time dynamics consists of moving down the Jordan block (and rightward along the chain) until, at a time of order $L$, it develops overlap with the Jordan block eigenvector\footnote{In this case, the eigenvector is $|L-1\rangle - |L\rangle$, which has eigenvalue $-\gamma_R$.} and reaches the steady state at a rate set by the corresponding eigenvalue ($\gamma_R$)  (which is also the spectral gap). Thus the gap sets the timescale for the relaxation of states at very late times, or alternatively, states that are already close to the right end of the chain (equivalently far down the Jordan block).

Away from the limit, the transition matrix is technically diagonalizable at any fixed $L$ \cite{Ueda20}. However, all (right) eigenvectors are concentrated in a region of size set by the correlation length $\xi \sim 1/|\log(|\gamma_R/\gamma_L|)$; thus, even though they are linearly independent, they overlap strongly with one another. Writing an initial state localized at the wrong end of the system in terms of these eigenvectors requires exponentially large coefficients in $L$, i.e., in terms of the right eigenbasis, the initial state can be written as $\rho_{\mathrm{init.}} = \rho_0 + \sum_i c_i e^{L/\xi} \rho_i$, where $\rho_0$ is the steady state, $\rho_i$ are the decaying right eigenvectors and $c_i$ are $O(1)$ coefficients. As time passes, the coefficient of a generic right eigenvector evolves as $\sim \exp(L/\xi - \Delta t)$. On timescales $L \sim t$, these coefficients become small, so the probability distribution approaches the steady state (consistent with the light cone bounds). It is only on timescales longer than this, when the particle nears the right end of the system, that relaxation is governed by the gap.

Thus the spectral gap does not govern all relaxation timescales. A quantity that is better tailored to describe relaxation is the size of the spectral resolvent
\begin{equation}\label{eq:resolvent}
    R(z) \equiv Q(z 1\!\!1 - M)^{-1}Q,
\end{equation}
where $1\!\!1$ is the identity matrix and $Q$ is the map that annihilates the steady state, and exactly preserves all other (generalized) eigenvectors of $M$. To see why, note that sandwiching Eq.~\eqref{eq:resolvent} between a norm $1$ operator $o$ and some state $\omega$ can be rewritten in real-time as 
\begin{equation*}
\braket{o|R(z)|\omega} = \int^{\infty}_0 dt e^{-zt} (\braket{o(t)}_\omega - \braket{o}_\infty),    
\end{equation*}
where $\braket{o}_\infty$ is the expectation in the steady state that $\omega$ approaches as $t\to\infty$. The resolvent is thus the Laplace transform of the real time relaxation of $o$ quenching from state $\omega$, and so it encodes relaxation time scales. In particular, at $z=0$ the value of the above integral is determined entirely by the time scale it takes for $\omega$ to relax; thus, the existence of a long-lived transient state $\omega$ has to manifest in a large norm for $R(0)$. For the present model we can choose $o,\omega$ such that this expression is $\langle o |R(0)|\omega\rangle = O(L)$\footnote{Choose $\omega$ to be the state with charge located on the first site, and $o=1-n_L$, one minus the occupation on the last site.}.

We can use this observation to define a notion of gap that \textit{does} correspond to the longest time-scale in the problem. Recall that the spectrum of a matrix $M$ is the set of points $z \in \mathbb{C}$ on which its resolvent is singular. The \emph{pseudospectrum}~\cite{trefethen2005spectra} is, instead, the set on which the resolvent is ``very large'': i.e., the $\epsilon$-pseudospectrum is the set of points $\{ z \in \mathbb{C} \, : \, \Vert R(z) \Vert \geq 1/\epsilon \}$, where $\Vert \cdot \Vert$ denotes the norm. For a finite-size matrix and small enough $\epsilon$, the pseudospectrum eventually consists of small disjoint sets around each eigenvalue. However, the limits of small $\epsilon$ and large $L$ can fail to commute, as they do in ASEP: taking the large-$L$ limit at finite $\epsilon$, one arrives at a pseudospectrum that is \emph{gapless} around $z=0$. In the maximally biased limit of ASEP mentioned above for appropriate observable and state $o,\omega$,  we get an exact expression $\langle o |R(z)|\omega\rangle=\frac{1}{z}((1-z/\gamma_R)^{1-L} - 1 )$ which indeed diverges if we take $L\rightarrow \infty$ first and $z\rightarrow0$ second, indicating TASEP has a gapless pseudospectrum. 

When the pseudospectrum is gapless, an immediate consequence~\cite{trefethen1991pseudospectra} is that the steady-state observables (and similarly the spectral gap) are unstable to perturbations. To see this, consider perturbing the channel from $M$ to $M+V$, where $V$ is small. We can develop a perturbative expansion for the expectation value of an observable $o$ in the steady state of the perturbed dynamics. A typical high-order term in the perturbation series takes the form
\begin{equation}\label{eq:resolvent_pertb}
\langle o | (RV)^n | \rho\rangle,
\end{equation}
where $\ket{\rho}$ is the unperturbed steady state and $R$ is the resolvent of the unperturbed channel at $z=0$. As we have just seen, $R$ can be singular near the origin in the thermodynamic limit, even if $M$ is gapped, so there is no particular reason for the $n$-th order perturbative expression to be small in $n$. We thus see that long lifetimes and instability to perturbations have a common origin, and they both have to do with the pseudospectrum, rather than the actual spectral gap. This should be contrasted with the perturbation theory for the ground state of a Hamiltonian. The corresponding high-order perturbative terms in that case takes the similar form
\begin{equation}\label{eq:tasep_resolvent}
\sum_{k \leq n} \langle 0 | (V R')^k o (R' V)^{n-k} |0\rangle.
\end{equation}
where $R'$ is now the resolvent of the Hamiltonian. For a Hermitian matrix, the resolvent is well-behaved away from the true spectrum, so terms in $n$-th order perturbation theory are suppressed by powers of $(\Vert V \Vert / \Delta)^n$ where $\Delta$ is the spectral gap. 

In the single-particle ASEP example, there is a natural perturbation that leads to sudden changes in observables and closes the spectral gap: it involves re-instating the coupling between the first and last site of the two ends of the system by a weak link of strength $\delta$ (i.e., allowing hopping between $L$ and $1$ with two possibly distinct rates both proportional to $\delta$). The spectral gap closes in the thermodynamic limit for any $\delta > 0$, and the steady states qualitatively change their character: instead of being localized at the edges, they become current-carrying and delocalized. Note that we only added a \emph{single} term to the dynamics: in the equilibrium case, the LPPL principle (and elementary perturbation theory) would have guaranteed the smooth evolution of local expectation values far from the perturbation in this case. Thus the example of ASEP illustrates that even the LPPL principle cannot be carried over directly to the nonequilibrium setting, much less the stronger property of stability against perturbations that act everywhere in the system. The key point about this perturbation is that it couples the steady state to transient states with extensive lifetime. This leads to an explosion of the resolvents in perturbation theory Eq.~\eqref{eq:resolvent_pertb} and an immediate breakdown of stability. 

\subsection{General lessons from ASEP}

The instability in ASEP originated from perturbations that induce a transition from the steady state to a state whose lifetime diverges with $L$. Indeed, it is quite clear that if such states exist, then coupling to them will always lead to a similar instability: under the perturbed dynamics, the steady state ``leaks'' into these transient states at a finite rate, while the rate of returning to the steady state becomes zero in the thermodynamic limit. To make this more precise, let $\rho_\infty$ be the steady state in question and $\rho_0$ some state that takes an extensive time\footnote{By extensive, we mean $O(L)$ with $L$ the linear system size. More generally, it is probably sufficient to assume that the time scale diverges with some positive power of $L$.} to relax to $\rho_\infty$. We modify our original channel, $\mathcal{E}_0$ by adding the perturbation $\mathcal{P}_\epsilon \rho = (1-\epsilon) \rho + \epsilon \mathrm{Tr}(\rho) \rho_0$, so that the full dynamics is generated by the map $\mathcal{E}_\epsilon = \mathcal{E}_0 \mathcal{P}_\epsilon$. Under repeated applications of this perturbed dynamics, the nontrivial steady state $\rho_\infty$ gets replaced by the trivial state $\rho_0$ at a rate $\epsilon$. By assumption, $\rho_0$ relaxes to $\rho_\infty$ at a rate that vanishes in the thermodynamic limit. Since the rate at which $\rho_\infty$ decays to the trivial state exceeds the rate at which it is repopulated, the steady state of the channel $\mathcal{E}_\epsilon$ changes qualitatively for any $\epsilon > 0$ in the thermodynamic limit.

This argument assumed that there exist long-lived transient states. In their absence, the resolvent should be bounded at $z=0$ and the steady state should be stable to perturbations, in agreement with the known stability of rapidly mixing channels~\cite{RapidMixing,Cubitt2015}. However, as noted in Sec.~\ref{sec:summary}, we do not expect such rapid mixing to be a property of non-trivial phases. Indeed, if we assume that $\rho_\infty$ in the above argument has some form of long-range correlations, and we choose $\rho_0$ to be short-range correlated, then long lifetime is ensured by the Lieb-Robinson bound. Thus, we can draw an important conclusion: \emph{any steady state in a nontrivial phase is unstable to certain arbitrarily weak perturbations}. This instability occurs even when the perturbations are much smaller in operator norm than the gap of the channel, in contrast to the situation in gapped equilibrium systems at zero temperature \cite{klich_gap}. 

The stability of non-trivial phases therefore has to originate from the fact that while long-lived transients might exist, allowed perturbations are unable to couple the steady state to them. In the ASEP example above, this happens if we forbid a coupling between the last site of the chain and any site an extensive distance away (measured anti-clockwise in Fig.~\ref{fig:tasep_fig}). In terms of the matrix in  Fig.~\ref{fig:tasep_fig}, this amounts to forbidding entries which couple the top and bottom of the Jordan block. When the couplings are restricted in this way the steady state is in fact stable. We conjecture that this principle holds more generally: gapped channels are stable to small perturbations as long as they are unable to couple the steady state to long-lived transient states of the unperturbed dynamics\footnote{Here, `unable to couple' means that the perturbations cannot produce an O(1) overlap between the steady state and the long-lived transient states under an O(1) time evolution.}. Just as in the ASEP case, this statement can be interpreted as bounding the effective sizes of resolvents as they appear in perturbation theory as in Eq.~\eqref{eq:tasep_resolvent} (see also Eq.~\eqref{eq:spectral_uniform} below). 

In the ASEP example, it is somewhat \emph{ad hoc} whether we consider the dangerous, end-to-end perturbations as physical or not; the coupling is local for the ring geometry in Fig.~\ref{fig:tasep_fig}, but is non-local for a linear chain geometry. In the problem of classifying many-body phases, which is our main concern in this paper, one naturally has some restrictions on the allowed set of perturbations: they have to be spatially local (furthermore, one might wish to enforce certain symmetry constraints as well). Thus, for a phase to be stable, it must be the case that none of the allowed perturbations can couple its steady states to long-lived excitations.  

This structure is particularly clear in the case of deterministic (irreversible) cellular automata (CA), which necessarily have Jordan block structures similar to the TASEP example we discussed above. A straightforward argument shows that the only eigenvalues the transition matrix of a CA can have are $1$ and $0$, thus they are all very strongly gapped. However, CA's generally have nontrivial relaxation times, which can sometimes be extremely long. Such relaxation times are \emph{entirely} due to the structure of their Jordan blocks, and when CA's are perturbed their stability under perturbations is determined precisely by whether these perturbations cause local or nonlocal moves along the Jordan blocks. Absolutely stable automata, such as Toom's rule, owe their stability to the fact they have Jordan blocks such that any \emph{spatially} local perturbation translates into a local move within the Jordan blocks (and thus an absence of couplings between steady states and long-lived transients).

In the remainder of the paper we will see how this structure plays out in concrete examples, and leads to stability or instability, depending on the structure of the dynamics and the set of perturbations we allow. In particular, in Sec.~\ref{sec:ToomLadder} we will encounter an example where although the perturbation cannot couple to long-lived states immediately, its repeated application does eventually lead to a diverging relaxation time and thus instability, albeit one that is invisible at any finite order in perturbation theory. With these examples in mind, we will return to the issue of formulating some general (conjectured) conditions for stability in Sec.~\ref{sec:stability}.

\section{Uniformity}\label{uniformity}

In this section we propose a sufficient condition for a stable open phase of matter. The intuition underlying our condition is as follows. Phases will be associated with connected regions in the space of local channels, along with (a subset of) their associated steady states. We call regions with these properties  ``steady-state bundles''. Open phases of matter are steady-state bundles which have a special property:  for nearby channels $\mathcal{E} ,\mathcal{E} '$ in the steady-state bundle, the steady states of  $\mathcal{E}'$  rapidly relax to those of $\mathcal{E}$ under evolution with $\mathcal{E}$. This technical condition is powerful, insofar as it implies that the correlation properties of steady states are analytic within the steady-state bundle. It also implies a condition on the resolvent of the channels (see Eq.~\eqref{eq:spectral_uniform} below), reminiscent of the gap used in defining ground-state phases of matter. We finish by comparing uniformity in the context of open systems to the well-known notion of gapped phase (encountered in Hamiltonian systems). 

\subsection{Definition of Uniformity}

We begin by defining steady-state bundles:

\begin{defn}
Let $D$ be a connected open subset of the space of local channels. Each channel $\mathcal{E}\in D$ is associated with a convex
space of normalised steady states $\mathcal{S}=S(\mathcal{E})$,
such that $(\mathcal{E}-\mathcal{I})(\mathcal{S})=0$\footnote{This could be generalised to include effective unitary dynamics on
the steady state space (e.g., in order to describe time crystals). This happens when the channel has eigenvalues with modulus $1$ but different from $1$.}, where $\mathcal{I}$ is the identity map. The elements of $\mathcal{S}$\emph{ }are called the \emph{phase
steady-states}, and they could be a proper subset of the set of all
steady states of a particular $\mathcal{E}$. The pair $(D,S)$ is a \emph{steady-state bundle}. \label{def:Steady-state-bundle-.}
\end{defn}

One thing that Def.~\ref{def:Steady-state-bundle-.} lacks is any
notion of a connection between different fibres, or indeed any guarantee
that $S$ is smooth/continuous. In the analogous problem of Hamiltonians and their ground states, this smoothness
is guaranteed by a gap. In our open-systems case, the same will be
guaranteed by our fast relaxation condition which we encapsulate in
the following definition.
\begin{defn}
A steady-state bundle $(D,S)$ is \emph{uniform} if it can be completely
covered by a set of balls $\{B_{i}\}_{i}$ such that: for any $\mathcal{E},\mathcal{E}'\in B_{i}$,
and for any $\rho\in\mathcal{S}=S(\mathcal{E}),$ there exists a
$\rho'\in\mathcal{S}'=S(\mathcal{E}')$ such that $\bigl\Vert\mathcal{E}^{t}\rho'-\rho\bigr\Vert\leq e^{-O(\Delta t)}$
where $\Delta$ does not scale with system size and is uniformly bounded
over all the $B_{i}$\footnote{We are implicitly working in the thermodynamic limit here. If one was working with finite systems and taking the TDL, it would
be important to specify that the radius of the balls should not go
to zero as the system size increases.}. \label{def:uniformity}
\end{defn}
In other words, a steady state bundle is uniform iff, whenever a channel has a phase steady state, that steady state arises from the finite
time evolution of a phase steady state of a nearby channel.
The reader might rightly worry that the statement $\bigl\Vert\mathcal{E}^{t}\rho'-\rho\bigr\Vert\leq e^{-O(\Delta t)}$ is ambiguous. What norm is being used? Will strict exponential decay be required in what follows? We will delay these details to App.~\ref{app:uniformity}. Regarding the norm, the statement of uniformity roughly says that $|\langle o|\mathcal{E}^{t}\rho'-\rho\rangle|\leq e^{-O(\Delta t)}\Vert o \Vert_\infty$ for local observables $o$. Regarding the need for exponential decay, for Theorems \ref{thm:LPPL} and \ref{thm:Uniformity-implies-analyticity} a sufficiently strong power law decay is sufficient. We will assume exponential decay only because it simplifies the proofs, although the ideas remain the same.

As a first result, we show that uniformity implies that the phase steady-state
spaces are isomorphic throughout $D$, in the sense that there is a one-to-one correspondence between states in $\mathcal{S}$ and $\mathcal{S}'$ for any $\mathcal{S},\mathcal{S}' \in S$. 
\begin{thm}\label{thm:isomorphism}
Uniformity implies the phase steady-states are isomorphic throughout
$D$. 
\end{thm}
\begin{proof}
Phase steady states are uniform throughout any of the open sets $B_{i}$
covering $D$. To see this, note that uniformity implies that for all $\mathcal{E},\mathcal{E}'\in B_{i}$,
we have $\lim_{t\rightarrow\infty}\mathcal{E}^{t}(\mathcal{S}')\supseteq\mathcal{S}$
but also $\lim_{t\rightarrow\infty}(\mathcal{E}')^{t}(\mathcal{S})\supseteq\mathcal{S}'$.
This implies an isomorphism $\mathcal{S}\sim\mathcal{S}'$ . Therefore
$S(\mathcal{E})$ isomorphic for all $\mathcal{E}\in B_{i}$. 

To see why this implies that the phase steady state spaces are uniform
throughout all of $D$, argue by contradiction. Suppose therefore
that there is a pair $\mathcal{E},\mathcal{E}'\in D$ with $\mathcal{S}\nsim\mathcal{S}'$.
Form a path $\mathcal{E}_{\lambda}$ between $\mathcal{E},\mathcal{E}'$
lying entirely in $D$. Then the following infimum exists $\lambda_{*}\equiv\inf\{\lambda\in[0,1]:\mathcal{S}_{\lambda}\nsim\mathcal{S}\}$.
By the definition of $\lambda_{*}$, $\mathcal{E}_{\lambda_{*}}$
must be arbitrarily close both to models with steady state manifold
$\sim\mathcal{S}$ \emph{and }steady state manifold $\nsim\mathcal{S}$.
But $\mathcal{E}_{\lambda_{*}}$ is in some \emph{open }ball $B_{i}$
in the cover, on which the isomorphism class of $S$ is constant.
So we have a contradiction, as required.
\end{proof}

Let us close this section with a few comments about the definitions we have introduced. Firstly, one might wonder why we decided to define uniformity in the way we did, which is in some sense backwards: we have assumed that every $\rho$ has a nearby ``parent'' $\rho'$ that quickly relaxes to $\rho$. It might have seemed more natural to instead assume that $\rho$ itself relaxes quickly under the application of $\mathcal{E}'$. The reason for using our definition of uniformity is highlighted by the above proof of isomorphism in $D$: had we used the other definition, the argument would not go through the same way. Intuitively this is related to the issue of fundamental asymmetry in quantum channels that is absent in the unitary case: it is possible to destroy long-range order in finite time it is not possible to create it. Thus, requiring that every $\rho$ is ``downstream'' from some other state nearby in the bundle is a stronger condition that allows us to establish more properties (indeed, we will rely on it again when we prove robustness of long-range order in the next section). On the other hand, some of our proofs below (such as those of LPPL and analyticity of local expectation values) would go through with both definitions and indeed it is possible that the two ways of defining uniformity might ultimately turn out to be equivalent. 

Finally, the use of the word ``bundle'' in Def.~\ref{def:Steady-state-bundle-.} might make our readers wonder whether these objects can possess a non-trivial topology, similarly to what occurs when, for example, one considers ground states in the quantum Hall effect \cite{simon_bundle}. Indeed, in App.~\ref{app:bundle} we provide a simple example of such a topologically non-trivial steady state bundle. We leave a more thorough exploration of such topological phenomena in steady state phases for future work.

\subsection{Uniformity implies analyticity, LPPL and the stability of long-ranged correlations.}\label{sec:uniformity_implies}

Uniformity implies that $(\mathcal{E})^{t}(\mathcal{S}')$
approaches $\mathcal{S}$ exponentially quickly (and vice versa) with
uniform rate $\Delta$. This has important consequences which we describe in this section.

Firstly, we show local perturbations to a channel in in a uniform steady-state bundle only lead to changes in the steady state local to those perturbations. In other words ``local perturbations perturb locally'' (LPPL). 
\begin{thm}
Within a uniform steady state bundle, local perturbations perturb locally. \label{thm:LPPL}
\end{thm}

\begin{proof}
Let $\mathcal{E} \in D$ for some steady state bundle $(D,S)$ and let $\mathcal{E}' \in D$ be the result of perturbing $\mathcal{E}$ at
location $y$. We denote their respective steady state manifolds by $\mathcal{S},\mathcal{S}' \in S$. By decomposing the circuits defining $\mathcal{E},\mathcal{E}'$,
we may write $\mathcal{E}=\mathcal{E}_{c}\mathcal{E}_{y}$ and
$\mathcal{E}'=\mathcal{E}_{c}\mathcal{E}'_{y}$ where $\mathcal{E}_{y},\mathcal{E}'_{y}$
are local channels supported on a neighbourhood of $y$, and $\mathcal{E}_{c}$
is a local channel. 

Suppose the perturbation is sufficiently small such that $\mathcal{E},\mathcal{E}'\in B_{i}$
for some $i$ -- i.e., they are both in one of the open sets used to
define uniformity. Take any $\rho\in\mathcal{S}$. Let $\rho'$ denote
the corresponding state in $\mathcal{S}'$ which quickly relaxes to
$\rho$ (which exists by uniformity). Use the exact identity
\begin{equation}
\mathcal{E}^{t}=(\mathcal{E}')^{t}+\sum_{s=0}^{t-1}(\mathcal{E})^{s}[\mathcal{E}-\mathcal{E}'](\mathcal{E}')^{t-s}\label{eq:KKp}
\end{equation}
relating time evolution under $\mathcal{E},\mathcal{E}'$, to evaluate any unit norm observable
$o_{x}$ at location $x$ 
\begin{equation}
\braket{ o_{x}|(\mathcal{E})^{t}|\rho'}=\braket{ o_{x}|\rho'}+\sum_{s=0}^{t-1}\braket{ o_{x}|(\mathcal{E})^{s}[\mathcal{E}-\mathcal{E}']|\rho'}. \label{eq:lppl_master}   
\end{equation}

Uniformity implies that the LHS approaches $\braket{o_{x}|\rho}$
exponentially quickly so that
\begin{equation*}
\braket{o_{x}|\rho}+e^{-O(\Delta t)}=\braket{ o_{x}|\rho'}+\sum_{s=0}^{t-1}\underbrace{\braket{o_{x}|(\mathcal{E})^{s}[\mathcal{E}-\mathcal{E}']|\rho'}}_{\lozenge_{s}}.
\end{equation*}
Two observations about the integrand $\lozenge_{s}$. Firstly, it
tends to zero exponentially quickly in $s$, as can be seen from writing.
\begin{align*}
\lozenge_{s} &=\braket{o_{x}|\mathcal{E}^{s+1}-\mathcal{E}^{s}|\rho'}
 \\ &=\braket{o_{x}|\rho}-\braket{ o_{x}|\rho}+e^{-O(\Delta s)}
  =e^{-O(\Delta s)}.
\end{align*}
Therefore the integrand may be cutoff at $t_{\epsilon}=\Delta^{-1}\log(\epsilon^{-1})$
while incurring error at most $\epsilon$. We assume $t>t_{\epsilon}$
henceforth and write
\begin{equation*}
\braket{o_{x}|\rho}+e^{-O(\Delta t)}=\braket{ o_{x}|\rho'}+\sum_{s=0}^{t_{\epsilon}}\lozenge_{s}+O(\epsilon). 
\end{equation*}

Secondly, we can bound $\lozenge_{s}$ using the Lieb-Robinson bound. where we introduced $r = |x-y|$. Introducing $r \equiv |x-y|$ we can write
\begin{equation}\label{eq:LPPL_LR}
|\lozenge_{s}| =\braket{o_{x}|(\mathcal{E})^{s}\mathcal{E}_{c}[\mathcal{E}_{y}-\mathcal{E}'_{y}]|\rho'}
\leq e^{\alpha(s-r/v)},
\end{equation}
(This follows from the fact that $\bra{o_x} \mathcal{E}^{s} \mathcal{E}_c$ is supported within $[x-vs,x+vs]$ up to exponentially small tails.) Now let $r >2vt_{\epsilon}$, such that the bound in Eq.~\eqref{eq:LPPL_LR} is guaranteed to be exponentially small for all $s \leq t_\epsilon$. In this case,
\begin{equation*}
\braket{o_{x}|\rho-\rho'}=e^{-O(\Delta t)}+t_{\epsilon}e^{-\alpha\frac{r}{2v}}+O(\epsilon).   
\end{equation*}
Now let us choose $\epsilon=e^{-ar}$ for some $a>0$. That implies $t_{\epsilon}=ar/\Delta$.
This is consistent with our use of the Lieb-Robinson bound provided
$r/(2v)>ar/\Delta$, which holds provided $a<\Delta/(2v)$. It
suffices to pick $a=\Delta/(4v)$. Putting all of these together, and
using $t>t_{\epsilon}$ gives
\begin{equation*}
|\braket{o_{x}|\rho-\rho'}|= Ce^{-r\Delta/(4v)}+C'(r/4v)e^{-\alpha\frac{r}{2v}}+C''e^{-r\Delta/(4v)},    
\end{equation*}
which is clearly exponentially decaying in $r$. 
\end{proof}

We are able to show more, namely that said correlation functions are analytic within a uniform steady state bundle, even for perturbations that act everywhere in the system. 
\begin{thm}
Uniformity implies analyticity of local correlations.\label{thm:Uniformity-implies-analyticity}
\end{thm}
\begin{proof}
It suffices to show that the properties of steady states are analytic
in any of the balls $B_{i}$ comprising our uniform bundle. Let $\mathcal{E}_{0}\in B_{i}$
be an arbitrary point in any of the balls, and let $\omega_{0}\in S(\mathcal{E}_{0})$
be any of the phase steady states. Consider any family of channels
in this ball emanating from $\mathcal{E}_{0}$, namely $\mathcal{E}_{\epsilon}\in B_{i}$
for $\epsilon\in[0,1]$. We can associate a steady state $\omega_{\epsilon}=\omega_{\epsilon}(\mathcal{E}_{\epsilon};\mathcal{E}_{0},\omega_{0})\in S(\mathcal{E}_{\epsilon})$
to each of these channels such that $\bigl\Vert\mathcal{E}_{\epsilon}^{t}\omega_{0}-\omega_{\epsilon}\bigr\Vert\leq e^{-O(\Delta t)}$.
$\omega_{\epsilon}$ are a one-dimensional family of states in our
steady-state bundle that all ``roughly look like $\omega_{0}$''. Let $o$
be a local observable. Our goal is to show that $\braket{o}_{\epsilon}\equiv\braket{o|\omega_{\epsilon}}$
is an analytic function of $\epsilon$. 

First, uniformity implies that $\braket{o}_{\epsilon} =\braket{ o|\mathcal{E}_{\epsilon}^{t}\omega_{0}} +e^{-O(\Delta t)}$. Now note that $o(\epsilon;t)\equiv\braket{ o|\mathcal{E}_{\epsilon}^{t}\omega_{0}}$
is a sequence of bounded analytic functions in $\epsilon$ which tend
to $\braket{o}_{\epsilon}$ \emph{uniformly} as $t\rightarrow\infty$. The uniform convergence follows directly from the fact that the exponential  convergence in $t$, with a fixed rate $\Delta$, holds for any two channels in the same ball, as required by Def.~\ref{def:uniformity}. $o(\epsilon;t)$ are bounded because $o$ is a bounded observable and they are analytic because they can be computed (up to exponentially small errors) by
performing a quantum mechanical simulation on a finite system with
size of order $t$ (due to the Lieb-Robinson bound), which is manifestly analytic in the matrix $\mathcal{E}_{\epsilon}$
and hence in $\epsilon$. Noting that the uniform limit of a family
of bounded analytic functions is analytic, we have that $\left\langle o\right\rangle _{\epsilon}=\lim_{t\rightarrow\infty}o(\epsilon;t)$
is analytic in $\epsilon$.
\end{proof}
Theorem~\ref{thm:isomorphism} showed that the steady state spaces are of the same dimension throughout the uniform region. But a stronger result is true. If, for example, the steady state $\rho$ of some channel $\mathcal{E}$ in the uniform region has long-range order, it follows that any other $\mathcal{E}'$ in the uniform region also has a steady state with long-range order. This is a consequence of fact that each steady state space can be obtained from the other to good approximation through a finite sequence of local channels. 
\begin{thm}
    Long-range order is robust throughout a uniform region\label{eq:LRO_theorem}
\end{thm}

\begin{proof}
To begin, we show that if $\mathcal{E}$ in a uniform bundle has a long-range ordered steady state $\rho$, then  any $\mathcal{E}'$ in the neighbourhood of $\mathcal{E}$ also has a long-range ordered steady state. By uniformity, there must exist a $\rho'=\mathcal{S}(\mathcal{E}')$ such that $\mathcal{E}^t \rho'=\rho + O(e^{-\Delta t})$.

By definition of long-range order, there must exist two unit norm
observables such that the connected correlation function in $\rho$
does not decay to zero at large distances. 
\begin{equation}
\lim_{|x|\rightarrow\infty}\sup\bigl|\bigl\langle o_{x}o_{0}|\rho\bigr\rangle^{\mathrm{c}}\bigr|>D>0 
\end{equation}

for some $O(1)$ constant $D$. Plugging in our definition of uniformity,
we have that
\begin{align}
\bigl\langle o_{x}o_{0}|\mathcal{E}^{t}|\rho'\bigr\rangle & =\bigl\langle o_{x}o_{0}|\rho\bigr\rangle+O(e^{-\Delta t})\label{eq:LRO1}\\
\bigl\langle o_{0,x}|\mathcal{E}^{t}|\rho'\bigr\rangle & =\bigl\langle o_{0,x}|\rho\bigr\rangle+O(e^{-\Delta t})\label{eq:LRO2}
\end{align}
We focus on the expression $\bigl\langle o_{x}o_{0}|\mathcal{E}^{t}$.
Provided $|x|>2vt$ where $v$ is the Lieb-Robinson velocity, we may
write this approximately as $\bigl\langle\tilde{o}_{x}\tilde{o}_{0}|$
where $\tilde{o}_{0,x}=\langle o_{0,x}|\mathcal{E}^t$ are operators based in cones of size $O(vt)$
around positions $0,x$ respectively with only exponentially decaying corrections
in $x$ which we can safely ignore. These operators have norm at most one by the contractive property of channels. Combining Eqs.~\eqref{eq:LRO1} and~\eqref{eq:LRO2} gives
\begin{equation}\label{eq:LRO3}
\bigl|\bigl\langle\tilde{o}_{x}\tilde{o}_{0}|\rho'\bigr\rangle^{\mathrm{c}}-\bigl\langle o_{x}o_{0}|\rho\bigr\rangle^{\mathrm{c}}\bigr|<Ce^{-\Delta t}  
\end{equation}
provided $|x|>2vt$. The constant on the RHS does not depend on $x$\footnote{This becomes clear using the more detailed description of uniformity discussed in App.~\ref{app:uniformity}.}. Taking $t=\Delta^{-1}\log(\frac{2}{CD})$ and the $\lim\sup$ as $x\rightarrow\infty$,
we find $\lim_{|x|\rightarrow\infty}\sup\bigl|\bigl\langle\tilde{o}_{x}\tilde{o}_{0}|\rho'\bigr\rangle^{\mathrm{c}}\bigr|>D/2>0$
which proves that $\rho'$ has long-range order.

To show that long-range order persists throughout the entire uniform
region, the above argument need only be repeated on the finite set
of neighbourhoods connecting any two channels in the uniform region. 
\end{proof}

Thus far, we see that uniform phase bundles have many of the properties we expect from a phase of matter. The number of steady states
is fixed throughout the phase. Local perturbations perturb locally, and  local correlation functions are smooth throughout a uniform steady state bundle. This leads us to propose that the uniformity condition is a sufficient condition for the stability of open phases of matter. 

\subsection{Uniformity and the spectrum}
Our definition of phase of matter does not make direct reference to the spectrum of the channels within the phase. Nevertheless, we can show that uniform steady-state bundles have a feature analogous to a spectral gap, in the form of
a condition on the resolvent. Consider the setup of Theorem~\ref{thm:Uniformity-implies-analyticity},
and consider (for simplicity) the following one-parameter family of
channels:
\begin{align}
\mathcal{E}_{\epsilon} & \equiv\mathcal{E}_{0}\prod_{x}(1+\epsilon v_{x}).\label{eq:deformed}
\end{align}
Provided that $\epsilon$ is sufficiently small such that $\mathcal{E}_0$ and $\mathcal{E}_{\epsilon}$
lie in the same uniform region, Theorem~\ref{thm:Uniformity-implies-analyticity}
implies $\braket{o}_{\epsilon}\equiv\lim_{t\rightarrow\infty}\braket{ o|\mathcal{E}_{\epsilon}^{t}|\omega_{0}}$
is an analytic function of $\epsilon$. Consider the Taylor expansion
of this expression in $\epsilon$. Analyticity implies that the coefficients
have a nonvanishing radius of convergence. That in turn implies that
\begin{equation*}
\lim_{t\rightarrow\infty}\biggl|\sum_{x_{1,\ldots,n}}\sum_{t_{1,\ldots,n}}^{t}\braket{ o|\mathcal{E}_0^{t_{n}}v_{x_{n}}\ldots\mathcal{E}_0^{t_{1}}v_{x_{1}}|\omega_{0}}\biggr|=e^{O(n)}    
\end{equation*}
grows no faster than exponentially in $n$ and does not diverge with system size. This expression can be
repackaged as 
\begin{align}\label{eq:spectral_uniform}
\braket{o|\left((\mathcal{I}-\mathcal{E}_{0})^{-1}\mathcal{V}\right)^{n}|\omega_{0}} =e^{O(n)}& &
\mathcal{V} \equiv\sum_{x}v_{x}, 
\end{align}
where it is understood that $(\mathcal{I}-\mathcal{E}_{0})^{-1}\equiv \sum^\infty_{k=0}\mathcal{E}^k_{0}$. Eq.~\eqref{eq:spectral_uniform} is certainly evocative of a gap condition, as it would follow if  $\mathcal{E}_{0}$  were self-adjoint and gapped away from eigenvalue
$1$.  However $\mathcal{E}_{0}$ is not generally self-adjoint. As a consequence, as we saw in Sec.~\ref{sec:ASEP},  the presence of a gap need not imply a useful bound on resolvents. Indeed, a similar argument to one given in Sec.~\ref{sec:ASEP} shows that Eq.~\eqref{eq:spectral_uniform}  diverges with system size even for $n=1$ for ASEP on an open chain, even though ASEP is gapped.

In summary, we have shown is that if $\mathcal{E}_{0}$ is stable against all local perturbations (i.e., it sits inside a uniform steady state bundle that has finite volume in the space of all local channels), then our generalized gap condition Eq.~\eqref{eq:spectral_uniform} must hold for arbitrary local Lindbladians $\mathcal{V}$. Using ASEP as an example, we noted that the generalised gap condition does not follow from the channel $\mathcal{E}$ being gapped. 

Is there any connection then between the spectrum of $\mathcal{E}$ and our definition of uniformity? We conjecture that the answer is yes: For models that are uniform, the gap sets the relaxation time-scale occurring in the uniformity definition. Here we again draw on intuition gleaned from ASEP in Sec.~\ref{sec:ASEP} as to the physical interpretations of gaps in open systems.  There we observed a distinction between two types of initial states. Firstly, states that are initially localized far away from the steady state and take a long time to relax as they need to make their way through the system first. Secondly, states that are already close to the edge, and therefore have a significant overlap with the steady-state to begin with. These latter states relax exponentially quickly at a rate dictated by the gap. Thus, the gap sets the relaxation time-scale for states that are already near the steady state.

We can try to generalize this picture to the many-body systems we consider here: given a channel $\mathcal{E}$, we can distinguish between states that relax rapidly under $\mathcal{E}$ and those that do not. For the former, it is still reasonable to expect, that their relaxation rate will be set by the spectral gap (indeed, were $\mathcal{E}$ to be gapless, we would not expect \emph{exponentially} rapid relaxation to occur). Roughly speaking, uniformity requires that the steady states of nearby channels within the phase belong to the first category. Thus we expect that the time scale $\Delta$ appearing in Def.~\ref{def:uniformity} will in fact coincide with the gap, which we formulate as
\begin{conjecture}\label{conj:uniformity=gap}
    For two channels $\mathcal{E}, \mathcal{E}'$ within the same uniform steady state bundle, the exponential decay rate of $||\mathcal{E}^t \rho' - \rho||$ equals the spectral gap of $\mathcal{E}$. 
\end{conjecture}
Here, we define gap analogously to the Hamiltonian case, i.e. we take any states with eigenvalues exponentially close to $1$ to be part of the manifold of steady states and consider the next leading eigenvalue after these. Our conjecture states that this eigenvalue approaches a value $\lambda_* < 1$ in the thermodynamic limit, and that $-\log{\lambda_*} ||\mathcal{E}^t \rho' - \rho|| \propto e^{-\log(\lambda_*) t}$ asymptotically. We provide numerical evidence for this conjecture in Sec.~\ref{sec:numerics} where we consider a one-dimensional probabilistic cellular automaton model for which uniformity can be established rigorously.

\subsection{ Uniformity in a volume versus uniformity along a submanifold }\label{sec:compare_gap_uniformity}

\begin{figure}
\includegraphics[width=0.55\columnwidth]{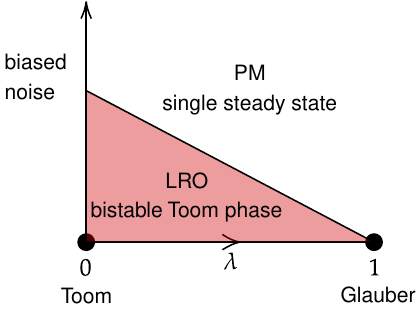}\caption{  A schematic phase diagram interpolating between 2D Toom ($\lambda=0$)
and (low-temperature, zero field) Glauber dynamics ($\lambda=1$).
Long-range order is unstable to biased noise for Glauber
dynamics, while Toom is stable. This implies that Toom/Glauber cannot
lie in a common uniform \emph{volume} of parameter space.\label{fig:toomglauber}}
\end{figure}

Def.~\ref{def:uniformity} implicitly involves a reference to some background manifold of channels in which the balls $B_i$ are to be drawn. This might be the manifold of all local channels or a submanifold restricted by some additional constraints, e.g., symmetries. We will refer to these two cases as \emph{uniformity in a volume}  and \emph{uniformity on a submanifold} respectively. An extreme limit of the latter be to look at only a only a one-dimensional manifold, i.e. a path, analogous to ``gapped paths'' of Hamiltonians. In fact, this is the version of uniformity we prove in a specific example (Stavskaya's cellular automaton) model below in Sec.~\ref{sec:stavskaya}.

While uniformity along a path is sufficient to prove the stability of steady state correlations along that path, it does not imply the stability of dynamical correlations, nor the stability of steady states to more general perturbations which take one away from the path. For example, consider a local evolution linearly interpolating
between the 2D Toom cellular automaton\footnote{For a definition and properties of this model, see e.g. Refs. \cite{GrinsteinToom} and \cite{ponselet_thesis}.} and zero-temperature (and zero
field) Glauber dynamics (Fig.~\ref{fig:toomglauber}); this could be achieved e.g. by randomly choosing between the two update rules at every site with some probabilities $\lambda$ and $1-\lambda$. 

Both of these models preserve fully polarized states, which will therefore be steady states everywhere along the path interpolating between them. Thus there is uniformity along the path connecting these two models, and this family of models exhibits long-range ordered steady-states. More generally, we expect Toom and Glauber to lie in the same uniform volume when restricting to the space of Ising symmetric channels. On the other hand, Toom is known to be inside a stable phase even if we break the symmetry restriction (e.g. by adding noise that is biased towards one of the two spin directions), unlike Glauber which only has a unique steady state once the symmetry is broken. Thus, if we work in the larger manifold of all local channels, the two are not in the same phase (and indeed, Glauber is not inside any stable phase at all). 

This difference is closely related to the different \emph{dynamical} properties of the two models. This can be quantified by considering the time needed to "erode" an island of flipped spins on top of a fully polarized state in the two cases. Glauber dynamics erodes such islands on a time scale quadratic in their linear size\footnote{This is due to the fact that in Glauber, the dynamics of domain walls is \emph{curvature-driven}, which is an inevitable consequence of the Z2 symmetry combined with detailed balance.}, while in Toom the erosion time is linear \cite{GrinsteinToom}. This difference explains why the latter can counter the effects of biased noise (which deterministically grows islands on top of the metastable state), while the former cannot. This also leads to the conjectured phase diagram\footnote{For $\lambda<1$, domains are still eroded on a linear time-scale, which implies stability to sufficiently weak noise.} in Fig.~\ref{fig:toomglauber} which shows Glauber as lying at the edge of a stable Toom phase. Thus, uniformity and the stability of steady-state correlations along a submanifold does not imply the stability of dynamics or the stability of correlations to more general perturbations which take one away from that submanifold.

A similar situation arises in Hamiltonian systems. For example~\cite{freedman2005line}
identify a smooth family of Hamiltonians, all of which have the \emph{same}
ground state. However, some of the Hamiltonians are gapped and stable,
while other are unstable and gapless (with associated slow dynamical
correlations). Thus,
knowing the ground state of a Hamiltonians does not determine whether
the underlying Hamiltonian is stable to perturbations, and
does not allow us to fully characterise temporal correlations or the nature of low-lying excitations. 

A more complete sufficient condition for stability of an open phase of matter would ideally imply the stability of both dynamical and steady-state correlations. It would also ideally be local in parameter space (like having a gap). Indeed it could be precisely a gap in some cases e.g., it might be that the spectral gap of the channel closes as one interpolates from Toom to Glauber, and that this heralds the sudden change in dynamical correlations. It is plausible that models which are in the same uniform \emph{volume} of parameter space (i.e., an absolutely stable phase) have qualitatively similar dynamics as well. in Sec.~\ref{sec:erosion} we conjecture a sufficient local condition for a channel to be within a such a phase, which combines dynamical properties (a generalization of the aforementioned 'fast erosion' property of Toom) with spectral features (a bound on the number of distinct steady states)

\section{Case study with uniformity: Stavskaya's model}\label{sec:stavskaya}

In this section we consider a one-dimensional cellular automaton with a nontrivial phase that is known to be stable under a certain class of perturbations. We prove that this family of models satisfies the uniformity criterion, which implies that the corresponding steady states have qualitatively the same correlations. The logic of our argument is fairly general, and we expect that similar proofs can be constructed in other perturbed automata. Having proved the uniformity criterion, we numerically explore the timescale that appears in this criterion and its relation to the inverse spectral gap. We find that these approximately track each other but do not coincide.

\subsection{Stavskaya's model}

Stavskaya's model is a 1D cellular automaton. Its configuration space is that of 1D binary strings $b_{x}=0,1$ for $x\in\{1,\ldots,L\}$. If $b_{x}=1$ for some $x$ we say there is an error at $x$. The update rule, implemented by dynamical map $K$, is
\begin{equation}
    b_x \rightarrow \min(b_x,b_{x+1})
\end{equation}
This rule has a simple graphical summary: Islands of errors, i.e., contiguous groups of sites with $b_x=1$, are eroded with velocity $1$ from the right. 
\begin{equation}
   \includegraphics[scale=1.5]{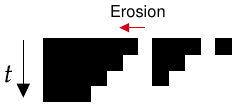}
\end{equation}
It is straightforward to verify that Stavskaya has precisely two steady states with periodic boundary conditions. 1) The  error-free state: $b_{x}=0$ $\forall x$, or $\vec{0}$; and 2) The error-full state $b_{x}=1$ $\forall x$ or $\vec{1}$. Therefore the $\lambda=1$ eigenvalue is doubly degenerate.

These two steady states remain stable in the presence of weak maximally biased noise, i.e., perturbations that flip $0 \to 1$ but not $1 \to 0$. This gives a perturbed Markov map
\begin{equation}
    K_\epsilon \equiv K \prod_{x}(1+\epsilon v_x)
\end{equation}
where $v_x=\sigma_x^+-\sigma_x^-\sigma_x^+$ and $\epsilon$ parameterises the noise strength. There is a transition occurring at $\epsilon = \epsilon_c \approx 0.3$ \cite{ponselet_thesis}. For $\epsilon < \epsilon_c$, there are two steady states in the thermodynamic limit: the state $\vec{1}$ which is an exact steady state of $K_\epsilon$, and another state which evolves continuously from $\vec{0}$ (i.e., it has $\braket{b_x} \ll 1$ when $\epsilon \ll 1$). In a finite system of size $L$, the latter is not a true steady state, as rare fluctuations will turn it into $\vec{1}$ on a timescale that is $O(e^L)$. At $\epsilon=\epsilon_c$ the two steady states coalesce and the system undergoes a phase transition of the directed percolation type, to a trivial high noise phase where only the $\vec{1}$ steady state remains. The existence of the bistable phase at small $\epsilon$ for maximally up-biased noise was proved some time ago (see \cite{ponselet_thesis} for a review). Here we will review one of those arguments, as it gives valuable intuition for when and why CA models are stable, which we will be able to relate to the notion of uniformity developed in the previous section. 

\subsection{Contour argument for stability of Stavskaya}\label{sec:contour}

\begin{figure}
    \centering
    \includegraphics[width=0.95\columnwidth]{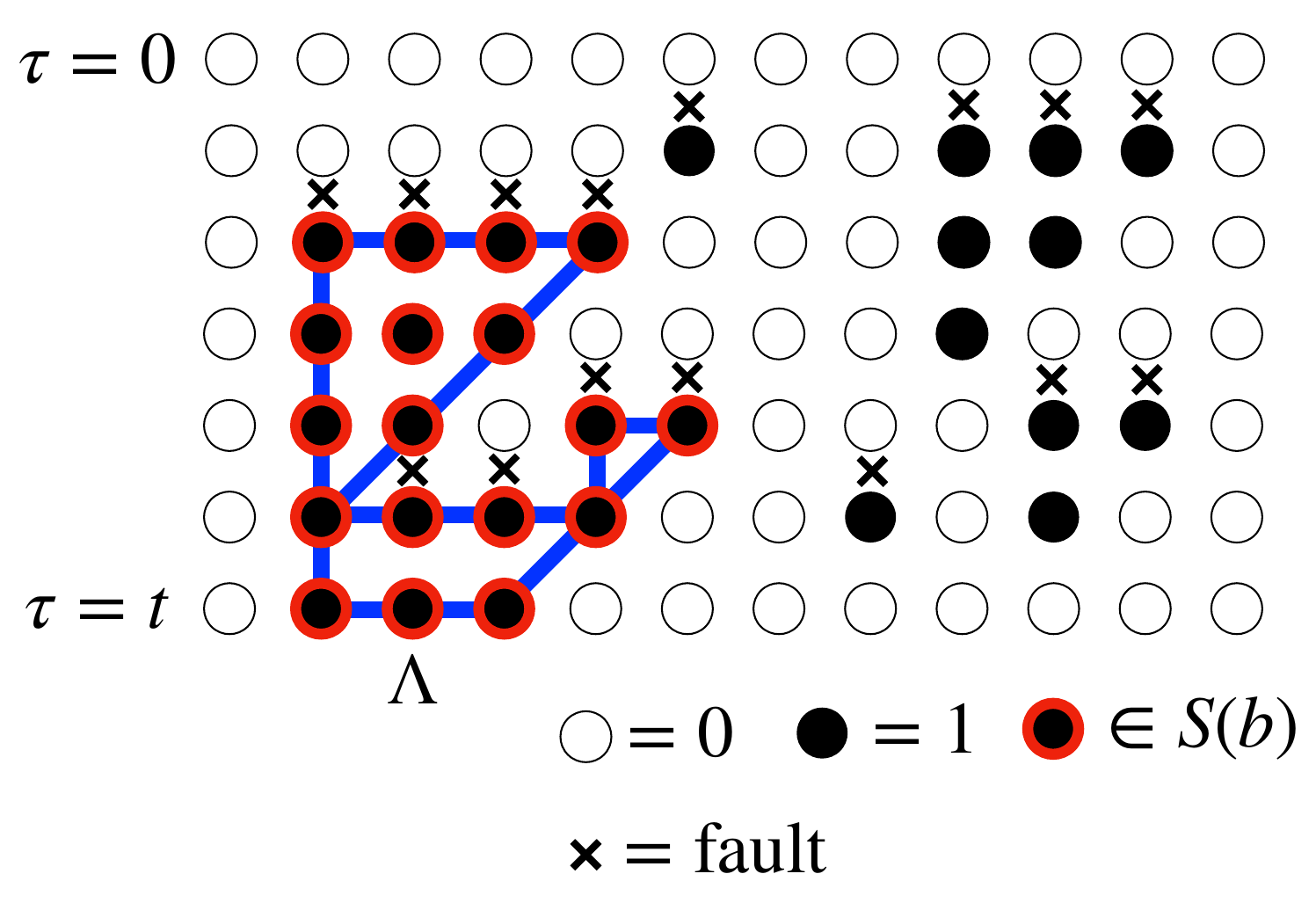}
    \caption{Stability of Stavskaya's model. The figure shows an example of a space-time configuration $\mathbf{b}$ terminating on set $\Lambda$. Errors (sites with bit value $1$) are denoted by black dotes. The set $S(\mathbf{b})$ is denoted by those black dots that have a read circle around them. The "x" symbols denote the occurrence of a fault i.e., a deviation from deterministic Stavskaya. The blue curve $\mathcal{C}$ enclosing $S(\mathbf{b})$ is the "contour". Note that horizontal edges of this contour are always associated with faults. Note there are also errors and faults that occur well away from $S(\mathbf{b})$; these are irrelevant for explaining the presence of errors in $\Lambda$ at time $t$.}
    \label{fig:contour_stavskaya}
\end{figure}

Here we summarise (and slightly extend) the ``contour argument''
for the perturbed Stavskaya model from Ref. \cite{ponselet_thesis}. This proof is a kind of spacetime version of Peierls' argument for the stability of the low temperature phase in the two-dimensional Ising model. To wit, we consider a scenario where the system is initialized in the error-free $\vec{0}$ state and evolved with Stavskaya dynamics subject to weak up-biased noise, below the threshold $\epsilon<\epsilon_c$. We the want to show that the probability of finding an error at a later time remains bounded by a time-independent quantity that is small when $\epsilon$ is small. In fact, we will show a stronger statement: the probability of having errors on sites precisely on set $\Lambda$ is bounded as
\begin{equation}
\mathbb{P}(b_{t}=\vec{\Lambda})=\braket{\vec{\Lambda}|K_{\epsilon}^{t}|\vec{0}}\leq C\epsilon^{|\Lambda|},\label{eq:ponselet_bound}
\end{equation}
where $|\Lambda|$ is the number of sites in set $\Lambda$, and $\bra{\vec{\Lambda}}$ denotes the spin configuration with errors precisely on $\Lambda$. It is important
that we work here in the thermodynamic limit, where system size has
been taken to infinity first. We first review the argument from Ref. \cite{ponselet_thesis} that applies in the case when $\Lambda$
is a contiguous set, and then we generalize it slightly to show that the result
also holds for non-contiguous sets. Since the bound in Eq.~\eqref{eq:ponselet_bound} decays exponentially with $|\Lambda|$, the overall probability of any particular site having an error is also bounded by some $O(\epsilon)$ quantity at all times. 

The goal of this section is two-fold. First, the proof provides some intuition on why Stavskaya's model is robust (when restricted to maximally up-biased noise), which will be useful when we develop some conditions for uniformity and compare it to other cellular automata where robustness fails in Sec.~\ref{sec:erosion}. More relevant to the purposes of this section, we will be able to use Eq.~\eqref{eq:ponselet_bound}, in its generalized form that applies to non-contiguous sets, to argue that the set of channels $K_\epsilon$ with $\epsilon < 6^{-3}$ with their pair of steady states forms a uniform steady state bundle. 

To prove Eq.~\eqref{eq:ponselet_bound}, we first Trotterise the expression $\braket{\vec{\Lambda}|K_\epsilon^t|\vec{0}}$ by inserting
a complete set of states at each integer time-step. Eq. \eqref{eq:ponselet_bound}
may then be expressed as a sum, $\mathbb{P}(b_{t}=\vec{\Lambda})=\sum_{\mathbf{b}}\mathbb{P}(\mathbf{b})$, over histories of the spacetime configurations
$\mathbf{b}=\{b_{\tau}\}_{\tau\in\{0,\ldots t\}}$,
with the boundary conditions $b_{t}=\vec{\Lambda}$ and $b_{0}=\vec{0}$,
where $\mathbb{P}(\mathbf{b})$ is the probability of space-time history $\mathbf{b}$. 

To each space-time history $\mathbf{b}$ we will assign a set of space-time locations $S(\mathbf{b})$. The
contours in the contour argument will (roughly) be the boundary of
$S(\mathbf{b})$. The sites making up a time-slice of $S(\mathbf{b})$ at time $\tau$ are denoted $S(\mathbf{b},\tau)$. We take $S(\mathbf{b},t)=\{(x,t):x\in\Lambda\}$ and we define the remaining slices inductively, moving backwards in time, as follows. At each step, the sites site $x\in S(\mathbf{b},\tau+1)$
will have $b_{x,\tau+1}=1$. Under unperturbed Stavskaya
dynamics, this site has two parents at the previous time slice that also needed to host errors in order, which are $(x,\tau)$ and $(x+1,\tau)$. If both of them indeed have errors then add them to $S(\mathbf{b},\tau)$ (see Eq.~\eqref{eq:inductive_cases}(a)). If, on the other hand, one of $b_{x,\tau},b_{x+1,\tau}\neq1$ (e.g., (see Eq.~\eqref{eq:inductive_cases}(b)).
then the error at $x,\tau+1$ must have been due to biased-up noise
acting in the $\tau\rightarrow\tau+1$ timestep at site $x$; we call
such an even a ``fault''. In that case do nothing, and proceed to
the next site in $S(\mathbf{b},\tau+1)$ and repeat until no sites remain. 
\begin{equation}\label{eq:inductive_cases}
\includegraphics[width=0.9\columnwidth]{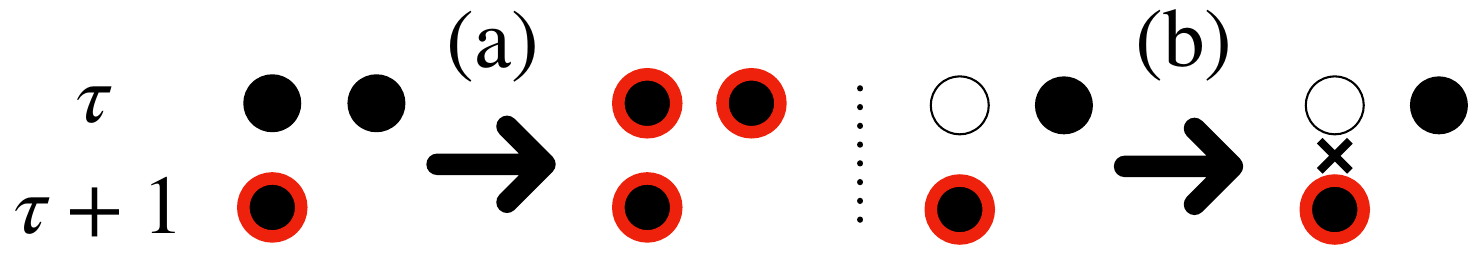}
\end{equation}
It is easy to verify that $S(\mathbf{b})$ is a connected spacetime set (the
sites added to $S$ during the inductive procedure always neighbour
some other site in $S$). A straightforward, but tedious calculation also shows that each set $S(\mathbf{b})$
can be unambiguously identified with a bounding contour $\mathcal{C}$
comprised of vertical (v), diagonal (d) $(1,-1)$, and horizontal
(h) edges only (see Fig.~\ref{fig:contour_stavskaya}). These edges never intersect a noise event. Moreover,
horizontal edges are associated with faults, e.g., a vertex is connected
to (one or two) horizontal edges if and only if it was caused directly by a
fault.
Therefore, every spacetime path can be associated with a set $S(\mathbf{b})$
and its corresponding contour $\mathcal{C}=\partial S(\mathbf{b})$, although
the mapping $S$ is not injective. To see that, consider the history
in Fig.~. Note for example that the errors in the spacetime history
that are not in $S(\mathbf{b})$ can be removed to form a $\mathbf{b}'$ with $S(\mathbf{b}')=S(\mathbf{b})$. 

Our contour has $n_{h}$ horizontal edges, $n_{v}$ vertical edges,
and $n_{d}$ diagonal edges. The fact the contour starts at the left
end-point of $\Lambda$ and ends at the right endpoint of $\Lambda$
places constraints on these edge numbers: we must have $n_{h}-n_{d}=|\Lambda|$
and therefore also $n_{h}\geq|\Lambda|$. The fact the contour starts
and ends at $\tau=t$ implies $n_{v}=n_{d}$. Therefore the total
number of edges is 
\begin{align}
n & =2n_{d}+n_{h}=3n_{h}-2|\Lambda|,\label{eq:nhnrelation}
\end{align}
and from $n_{h}\geq|\Lambda|$ we have that $n\geq|\Lambda|$. We can now replace the sum over space-time histories with a sum over contours, using the identity
$1=\sum_{\mathcal{C}}\delta(\mathcal{C}=\partial S(\mathbf{b}))$, to get
\begin{align*}
\braket{\vec{\Lambda}|K_{\epsilon}^{t}|\vec{0}} &= \sum_{\mathbf{b}} \mathbb{P}(\mathbf{b})
=\sum_{\mathcal{C}}\sum_{\mathbf{b}}\mathbb{P}(\mathbf{b})\delta(\mathcal{C}=\partial S(\mathbf{b})) \\ &\equiv \sum_{\mathcal{C}}\mathbb{P}(\mathcal{C}).
\end{align*}
Thus we find that the probability of having all up spins on set $\Lambda$
at time $t$ is the probability of having spacetime history with a
contour $\mathcal{C}$ that agrees with $\Lambda$ at time $t$.
In fact, we can readily bound $\mathbb{P}(\mathcal{C})$. To have
a contour $\mathcal{C}$, we must have a noise event associated to
each h-type edge. For the local noise we consider, such events are
exponentially suppressed in the number of faults $\mathbb{P}(\mathcal{C})\leq\epsilon^{n_{h}}=\epsilon^{(n+2|\Lambda|)/3}$\footnote{Note that this upper bound is rarely saturated, because obtaining a contour
$\mathcal{C}$ from $\mathbf{b}$ constrains the space-time configuration of
$\mathbf{b}$ outside of the contour $\mathcal{C}$, which will further suppress
the probability.} We arrange the sum over contours according to length
to obtain
\begin{align*}
\braket{\vec{\Lambda}| K_{\epsilon}^{t}|\vec{0}} &= \sum_{n\geq|\Lambda|}\sum_{\mathcal{C}:\mathrm{len}(\mathcal{C})=n}\mathbb{P}(\mathcal{C}) \\ &\leq\sum_{n\geq|\Lambda|}\epsilon^{(n+2|\Lambda|)/3}\sum_{\mathcal{C}:\mathrm{len}(\mathcal{C})=n} \\
&\leq\sum_{n\geq|\Lambda|}\epsilon^{(n+2|\Lambda|)/3}3^{n},
\end{align*}
where in the last step we used that the total number of contours of length
$n$ is upper-bounded by $3^{n}$ ($n$ edges, 3 possible
edge types). The RHS converges provided $\epsilon<3^{-3}$, in which case we find
\begin{equation}
    \braket{\vec{\Lambda}| K_{\epsilon}^{t}|\vec{0}}\leq\frac{(3\epsilon)^{|\Lambda|}}{1-3\epsilon^{1/3}},
\end{equation}
which shows an exponential suppression in $\Lambda$ as claimed. 

In fact, we now argue that this bound also holds when $\Lambda$ is a disjoint
union of intervals $\Lambda=\cup_{i}\Lambda_{i}$. For any space-time
history consistent with such a $\Lambda$, we can form $S(\mathbf{b})$ inductively
as before. The resulting set is a disjoint union of sets seeded by
each of the intervals $\Lambda_{i}$, namely $S(\mathbf{b})=\cup_{i}S_{i}(\mathbf{b})$.
This latter fact can be argued by contradiction. If there is
a point $x\in S_{i}(\mathbf{b})\cap S_{j}(\mathbf{b})$ for some $i \neq j$, then there is a path $p_{i}$
of errors in spacetime connecting $\Lambda_{i}$ to $x$ via a set of
vertical and diagonal edges, and similarly a path $p_{j}$. As $\Lambda_{i}\cap\Lambda_{j}=\emptyset$
(and wlog $\Lambda_{i}$ is to the left of $\Lambda_{j}$), there
must by some point between the two sets ( $\Lambda_{i}<y<\Lambda_{j}$
) such that $y\notin\Lambda$. Now, as $b_{y,t}=0$, one of $b_{y,t-1},b_{y+1,t-1}=0$.
Continuing back in time inductively, we can form a path $q$ of vertices
which are error free and which connects $(y,t)$ to some point in
the $\tau=0$ time slice. This path is comprised of v and d type edges only. We arrive
at a contradiction because $q$ must intersect $p_{i}\cup p_{j}$
at some point. Yet each of the vertices of $q$ are error free, while
those of $p_{i,j}$ all have errors. 

Since $S(\mathbf{b})=\cup_{i}S_{i}(\mathbf{b})$ is a disjoint union, each
of the sets $\cup_{i}S_{i}(\mathbf{b})$ has a contour $\mathcal{C}_{i}$ associated to it in the same way as before. The only complication
is that there are some correlations between the contours as they are
not allowed to intersect. We can nevertheless proceed to bound the error probability
as before:
\begin{align*}
\braket{\vec{\Lambda}| K_{\epsilon}^{t}|\vec{0}} &=\sum_{\mathbf{b}}\mathbb{P}(\mathbf{b}) =\sum_{\{\mathcal{C}_{i}\}}\sum_{\mathbf{b}}\mathbb{P}(\mathbf{b})\delta(\mathcal{C}_{i}=\partial S_{i}(\mathbf{b})) \\ &=\sum_{\{\mathcal{C}_i\}}\mathbb{P}(\{\mathcal{C}_{i}\}).
\end{align*}
We must now bound the probability of finding a particular collection of
contours $\{\mathcal{C}_{i}\}$. Each of these contours has $n_{h}$
horizontal edges, and so we must have at least $\sum_{i}n_{h,i}$
faults to achieve set $\{\mathcal{C}_{i}\}$. Therefore
\begin{equation*}
\braket{\vec{\Lambda}| K_{\epsilon}^{t}|\vec{0}}\leq\sum_{\{\mathcal{C}_{i}\}}\epsilon^{\sum_{i}n_{h,i}}.    
\end{equation*}
The sum has the implicit constraint that contours cannot overlap. However,
we can relax this constraint to obtain a further upper bound
\begin{equation*}
\braket{\vec{\Lambda}| K_{\epsilon}^{t}|\vec{0}}\leq\prod_{i}\sum_{\mathcal{C}_{i}}\epsilon^{n_{h,i}}.
\end{equation*}
We can now bound each of the terms in the product in the usual way
to obtain
\begin{equation*}
\braket{\vec{\Lambda}| K_{\epsilon}^{t}|\vec{0}}\leq\left(\frac{3\epsilon}{1-3\epsilon^{1/3}}\right)^{|\Lambda|}.    
\end{equation*}
(At an intermediate step we used the fact that the number of connected
components of $\Lambda$ is at most $\left|\Lambda\right|$). For
sufficiently small $\epsilon<6{}^{-3}$ we thus get
\begin{equation*}
\langle\vec{\Lambda}| K_{\epsilon}^{t}|\vec{0}\rangle\leq(6\epsilon)^{|\Lambda|}.    
\end{equation*}

\subsection{Uniformity of Stavskaya with up-biased noise}
\label{sec:stavskayauniform}
We can adapt the argument above to establish
a limited notion of uniformity amongst the biased-up perturbed Stavskaya
models, i.e., $K_{\epsilon}$ $\forall\epsilon<\epsilon_{c}$. We wish
to establish that nearby steady states in the ensemble above relax rapidly to one another. As all states in the above family have $\vec{1}$ as a steady state, we just need to establish that the ``mostly-down''
steady states rapidly relax to one another. In fact, we can establish
this rigorously drawing on results in \cite{ponselet_thesis,de_maere_ponselet},
particularly Theorem 11 of the former.

There it is shown that there exists an $\epsilon_{*}$ such that for
all $\epsilon<\epsilon_{*}<1$ the following holds. Suppose that we have a state (probability
distribution) $\omega$ on the configurations $\mathbf{b}$ which obeys $\omega(\mathbf{b} = \vec{\Lambda}) = \braket{\vec{\lambda}|\omega}\leq \kappa(\epsilon')^{|\Lambda|}$
for all regions $\Lambda$ and for some fixed $\kappa>0$ and $0\leq\epsilon'<\epsilon_{*}$
independent of $\Lambda$. Then there are constants $C$ and $\gamma > 0$
(independent of $\epsilon,\epsilon'$) such that
\begin{equation}
\bigl|\braket{o|K_{\epsilon}^{t}|\omega}-\braket{o|\rho_{\epsilon}}\bigr|<C || o||_{\infty} d_0 e^{-\gamma t},
\end{equation}
for all observables $o$. Here $\rho_{\epsilon}$ denotes the mostly
down steady state of $K_{\epsilon}$ and $d_{o}$ is the number of sites on which $o$ is supported.
This result applies not only to Stavskaya, but to an entire class of deterministic CA that satisfy the so-called monotonicity and erosion properties, which includes Toom's rule among others (see Ref. \cite{ponselet_thesis} for details). Intuitively, it says that if the initial state is sufficiently close to $\vec{0}$, then it relaxes exponentially quickly to the appropriate steady state of $K_\epsilon$. 

By the generalized version of the contour argument in the previous subsection, as long as $\epsilon < \epsilon_{**} \equiv 6^{-3}$, we have that $\braket{\vec{\Lambda}|\rho_{\epsilon'}}\leq C(\epsilon')^{|\Lambda|}$
for all finite subsets of the lattice, so that $\rho_\epsilon$. satisfies the condition of the above lemma. We thus have that for $\epsilon,\epsilon' \in D = [0,\min(\epsilon_{*},\epsilon_{**})]$, 
\begin{equation}\label{eq:0toepsrelaxation}
\bigl|\braket{o|K_{\epsilon}^{t}|\rho_{\epsilon'}}-\braket{o|\rho_{\epsilon}}\bigr|<C||o||_{\infty}d_{o} e^{-\gamma t},
\end{equation}
indicating that $\{K_\epsilon\}_{\epsilon \in D}$, together with the pair of steady states $(\rho_\epsilon,\vec{1})$, forms a uniform steady state bundle.  This argument shows that there exists at least one example of a non-trivial uniform steady-state bundle; namely the family consisting of Stavskaya dynamics with weak up-biased noise.\footnote{Note that Stavskaya is not "uniform in a volume" (Sec.~\ref{sec:compare_gap_uniformity}) because bistability is destroyed by any perturbation that can flip an up spin to a downs spin.}

 It is natural to ask whether the above proof of stability applies also to other cellular automata like Toom's rule. On an intuitive level, we expect the result to be true (indeed, our conditions for stability formulated in Sec.~\ref{sec:stability} bear some resemblance to the conditions needed to ensure stability in CA), there are difficulties in directly applying the same proof strategy in the more general case. In particular, it turns out that for some local noise, $\epsilon'$-perturbed Toom's rule fails to obey the required condition $\braket{\vec{\lambda}|\omega}\leq \kappa(\epsilon')^{|\Lambda|}$\footnote{See Theorem 9 of Ref.~\cite{ponselet_thesis} and Theorem 4.1 of Ref.~\cite{fernandeztoom}}, so Theorem 11 of Ref.~\cite{ponselet_thesis} cannot be applied directly to prove uniformity. Finding alternative routes to establish uniformity is an interesting problem for future work.

\subsection{Numerics}\label{sec:numerics}

In our definition of uniformity, Def.~\ref{def:uniformity}, there appears a decay rate $\Delta$, that sets how fast the steady states of the unperturbed dynamics relax to those of the perturbed one. Having argued analytically that a form of uniformity holds in Stavskaya's CA with biased noise, we now turn to the question of what sets this time scale in this model. As we now show, based on numerical simulations, the decay rate appears to agree with the spectral gap of the Markov matrix $K_\epsilon$ generating the noisy dynamics. This confirms the intuition outlined earlier in the paper: while the gap does not set the mixing time for arbitrary initial states, which might take an $O(L)$ time to reach a steady state, the gap \emph{does} correctly predict the time it takes to relax within some set of states that are sufficiently close to the steady state, to which the unperturbed steady states belong. 

We begin by considering the ``low-lying'' spectrum of the Markov generator $K_\epsilon$. We do this by performing exact diagonalization on $K_\epsilon$ for small chains of length $L \leq 14$ to compute the three largest in magnitude eigenvalues, which we denote $\lambda_{0,1,2}$. The leading eigenvector, corresponding to $\lambda_0 = 1$, is the steady state, which is unique when $\epsilon > 0$ at any finite $L$. We are interested in the behavior of $\lambda_{1,2}$ as a function of $\epsilon$ and $L$. To draw analogies with the spectra of Hamiltonians and Lindbladians, we will plot these as $-\log|\lambda_{1,2}|$.

We can distinguish two different regimes, depending on whether $\epsilon$ is above or below the critical value $\epsilon_c$ at which the second steady state of Stavskaya becomes unstable. Based on previous numerics~\cite{ponselet_thesis}, this is expected to happen at $\epsilon_c \approx 0.29$. For $\epsilon < \epsilon_c$, we observe that the first two eigenvalues are exponentially close to each other, i.e., $-\log|\lambda_1| = O(e^{-L})$, consistent with the exponentially long lifetime of the metastable steady state~\cite{ponselet_thesis,Gacs_on_Toom}. For $\epsilon > \epsilon_c$, on the other hand, we find that $-\log|\lambda_1|$ is no longer well fit by an exponential. Instead, the best fit we find is of the form $-\log|\lambda_1| = \Delta + c/L^2$, extrapolating to a finite constant $\Delta$ as $L\to\infty$, signalling that the steady state is now unique in the thermodynamic limit. This is shown by the lower curve in Fig.~\ref{fig:Gap_vs_relaxation}(a).

\begin{figure}
    \centering
    \includegraphics[width=1.\columnwidth]{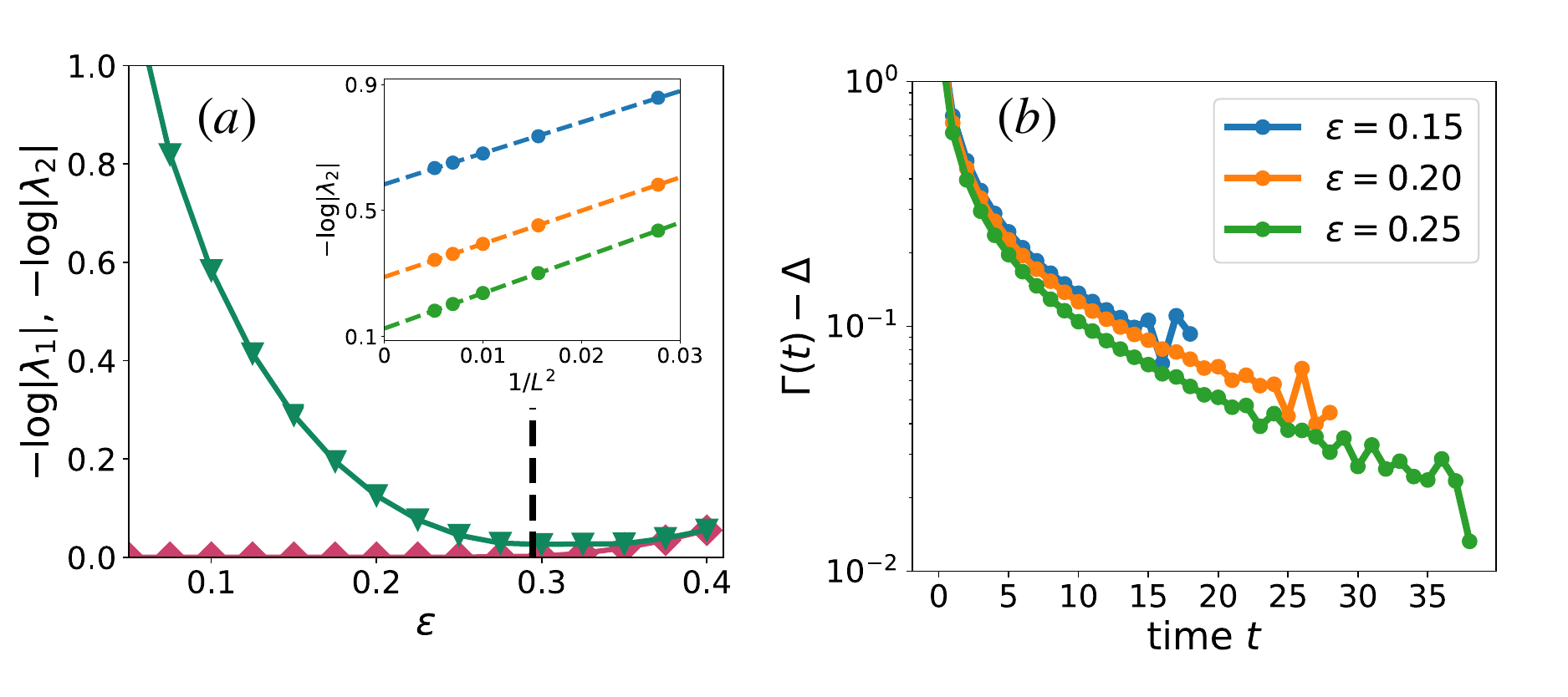}
    \caption{Gap and relaxation in the noisy Stavskaya dynamics. (a) shows the size of the first two subleading eigenvalues, $-\log|\lambda_{1,2}|$ For $\lambda_1$ and $\epsilon < \epsilon_c \approx 0.29$ (demarcated by the dashed vertical line), we use the eigenvalue for $L=14$ which is already very close to zero. For the other data points, we plot a value that has been extrapolated in $L$ using a fit of the form $\Delta + c/L^2$ (see inset). The data show a clear gap $\Delta > 0$ above the two steady state when $\epsilon < \epsilon_c$, which appears to close and reopen as we cross $\epsilon_c$. In (b) we plot the time-dependent decay rate, defined in Eq.~\eqref{eq:DefDecayRate} relative to the spectral gap extrapolated from the data in (a). Our data is consistent with $\lim_{t\to\infty}\Gamma(t) = \Delta$.}
    \label{fig:Gap_vs_relaxation}
\end{figure}

For the next leading eigenvalue, we find that the scaling $-\log|\lambda_2| = \Delta + c/L^2$ gives a good fit on both sides of the transition (the inset of Fig.~\ref{fig:Gap_vs_relaxation}(a) shows it for a few different values of $\epsilon < \epsilon_c$). Above the critical noise strength, we find that the value of $\Delta$ obtained from this fit approaches the one obtained from the same fit for $\lambda_1$, suggesting a continuum of states above the gap. In the non-trivial phase, we still find a gap $\Delta$ above the two steady states that remains finite in the thermodynamic limit. However, its value decreases gradually as $\epsilon\to\epsilon_c$ from below. While our finite size numerics are not sufficient to see the gap closing, we find that it plateaus at a small value, $\Delta < 0.03$ in the vicinity of the critical point, before starting to grow again as we get deeper into the trivial phase. We expect that in the thermodynamic limit, at $\epsilon = \epsilon_c$ the system has a gapless spectrum\footnote{We note that both analytical arguments~\cite{ponselet_thesis} and previous numerics~\cite{Stavskaya_DP} suggest that the critical point belongs to the directed percolation universality class.}. 

Thus, in the $\epsilon < \epsilon_c$ regime, there are two steady states, with eigenvalues exponentially close to one another, with the rest of the spectrum separated from them by a gap $\Delta$ that remains finite in the thermodynamic limit. On the other hand, as we have proven above, uniformity holds in some regime close to $\epsilon=0$, and in fact we expect it to extend all the way up to $\epsilon_c$. We now want to understand how these two features are related to each other. In particular: does the spectral gap determine the rate of decay that appears in Def. \ref{def:uniformity}?

To answer this question, we also need to extract the decay rate from our numerics. To do so, we define a time-dependent decay rate as follows:
\begin{equation}\label{eq:DefDecayRate}
   \Gamma(t) \equiv -\log\left(\frac{b_x(t+1)-b_x(t)}{b_x(t)-b_x(t-1)}\right),
\end{equation}
where $b_x(t)$ is the local magnetization at time $t$ on site $x$. If the magnetization decays exponentially, as it should when uniformity holds, then $\Gamma(t)$ will approach a constant at late times, which we can equate with the decay rate. In practice, we simulate many trajectories of the noisy dynamics starting from the $\ket{\vec{0}}$ state and average over them to get $b_x(t)$ (in fact, to further reduce fluctuations, we replace $b_x$ with its average $\sum_x b_x/L$). For $\epsilon < \epsilon_c$, the magnetization quickly approaches a constant $<1$ at late times, indicating that the system settles down to the appropriate perturbed steady state. We then plug the results into Eq.~\eqref{eq:DefDecayRate} to get the time-dependent decay rate.

The results are shown in Fig.~\ref{fig:Gap_vs_relaxation} (b), where we compare the decay rate $\Gamma(t)$ to the spectral gap $\Delta$ obtained from diagonalizing $K_\epsilon$.  We note that in order to get a smooth curves for $\Gamma(t)$ for sufficiently long times, we need to average over a very large number of realizations, especially for small $\epsilon$ where relaxation happens very quickly. (In Fig.~\ref{fig:Gap_vs_relaxation}(b) we used around $5 \times 10^9$ for $\epsilon=0.25$ and almost $2 \times 10^11$ for $\epsilon=0.15$). Even with such large sample sizes, the effects of noise are still clearly visible at the latest times; nevertheless, in the regime where the data appears to be well converged, it is consistent with the conjecture $\Gamma \equiv \lim_{t\to\infty}\Gamma(t) = \Delta$. We conjecture that this is true more generally: when uniformity holds in a phase, with exponential decay as in Def.~\ref{def:uniformity}, then the channels in that phase are gapped above their steady state manifold, and the gap is equal to the decay rate. 

\section{Stability criteria beyond uniformity}\label{sec:erosion}

The uniformity criterion we introduced in the previous section is powerful once it has been established, but it has the obvious drawback that proving it requires knowledge of the steady states and dynamics within an entire region of parameter space. Ideally, one would like to have a stability criterion that could be expressed entirely in terms of the properties of the \emph{unperturbed} dynamics---similarly to some existing stability results for Hamiltonians~\cite{michalakis}. In this section we explore the possibility of such a sufficient condition. To gain insight, we consider various kinds of known instabilities. We first consider an example where instability of a steady is apparent already at the perturbative level, and then show that such perturbative instabilities are absent in CA that obey a simple erosion property. We then consider another example that shows that even when this condition is satisfied, non-perturbative instabilities can appear. We then attempt to synthesize these examples to conjecture some conditions for stability that we expect to apply even beyond deterministic CA models. 

\subsection{Perturbative instability for Stavskaya with down-biased noise}\label{sec:StavskayaUnstable}

In Sec.~\ref{sec:contour} we have seen how Stavskaya's cellular automaton admits a stable phase when we perturbed it with maximally up-biased noise. We now consider the case of down-biased noise and show that the steady state $\ket{\vec{1}}$ is unstable in this case. We will show
that this is a ``perturbative instability'' i.e., it is associated
with an IR divergence (long-times/large system size) appearing at
a finite order in perturbation theory. The perturbative argument can be formulated by finding a generalised eigenbasis for $K_0$ and performing manipulations similar to those used in time independent perturbation theory in quantum mechanics. Here we opt for a slightly simpler (but equivalent) approach.

Starting with $|\vec{1}\rangle$, consider the effect on $b_{x}$
at long times in response to adding maximally down-biased noise. The generator
of dynamics is $K_{\epsilon}=K_0\prod_{x}(1+\epsilon v_{x})$ where
$v_{x}=\sigma_{x}^{-}-\sigma_{x}^{+}\sigma_{x}^{-}$ and $K_0$ is the
unperturbed Stavskaya dynamics. We again consider the expectation value $b_x(t)$, which can be interpreted
as the probability of having an error persist at location $x$ up
to time $t$:
\begin{equation}
p_{t}(\epsilon)\equiv\langle b_{x}| K_{\epsilon}^{t}|\vec{1}\rangle.\label{eq:stavskayaref}
\end{equation}
Note that $p_{t}(0)=1$. We will now argue that there is an instability in the sense that $\lim_{\epsilon\to0}\lim_{t\to\infty}p_t(\epsilon) \neq \lim_{t\to\infty}\lim_{\epsilon\to0}p_t(\epsilon)$: at long times $\ket{\vec{1}}$ evolves into a steady state that looks very different, even when $\epsilon$ is infinitesimally small (in fact, it evolves towards the unique steady state with $b_x = O(\epsilon)$). 

To see the instability, we expand $p_{t}(\epsilon)$
in $\epsilon$. The first order term takes the form $p_{t}'(0)\epsilon$, where
\begin{equation*}
p_{t}'(0)=\sum_{s=0}^{t-1}\sum_{y=1}^{L}\langle b_{x}| K_0^{s}v_{y}|\vec{1}\rangle.    
\end{equation*}
Here, $v_{y}|\vec{1}\rangle=|0_{y}\rangle-|\vec{1}\rangle$
where $|0_{x}\rangle$ is the configuration with a single $0$
at position $y$ in a sea of $1$'s. $\ket{\vec{1}}$ is invariant under $K_0^s$. On the other hand, when acting on $\ker{0_y}$, the deterministic Stavskaya
dynamics grows the single $0$ into a contiguous string with support
on set $[y-s,y]$, until it encompasses the entire system at $s = L-1$ (we assume periodic boundary conditions). We can summarize this as
\begin{equation*}
K_0^{s}v_{y}|\vec{1}\rangle=\begin{cases}
|0_{[y-s,y]}\rangle-|\vec{1}\rangle & s<L-1,\\
|\vec{0}\rangle-|\vec{1}\rangle & \mathrm{otherwise.}
\end{cases}
\end{equation*}
Taking the inner product with $\langle b_{x}|$ gives
\begin{align*}
p_{t}'(0) & =-\sum_{s=0}^{t-1}\sum_{y=1}^{L}I(x\in[y-s,y])\rightarrow\frac{1}{2}L(L+1)+L(t-L)
\end{align*}
where the final line obtains for $t\geq L$. This derivative plainly
diverges in the $t\rightarrow\infty$ limit, indicating that switching on an infinitesimal amount of unbiased noise can change the steady state magnetization by a large amount.

\subsection{Erosion implies perturbative stability}\label{sec:PertStab}

Summarizing what we have so far, Stavskaya's model is absolutely stable when we restrict to completely up-biased noise, and becomes unstable already at first order in perturbation theory for down-biased noise. In fact, it is easy to see that the latter result generalizes to anything other than the maximally up-biased case: as long as there is a non-zero chance that $1$'s can be flipped to $0$'s, the same perturbative instability occurs  (the steady state $\ket{\vec{0}}$, on the other hand, stays stable, by the same contour argument that we presented in Sec.~\ref{sec:contour}). The reasons for the instability are quite clear: the deterministic dynamics grows islands of $0$'s and thus it offers no protection for the $\ket{\vec{1}}$ state. 

The situation is different for the other steady state $\ket{\vec{0}}$, which is indeed stable (in the sense that the steady-state magnetization evolves continuously with $\epsilon$) for any local noise~\cite{ponselet_thesis}, biased or not. At a heuristic level, we can argue for this as follows. Turning on the perturbation introduces ``errors'' above the initial steady state with a density $\sim \epsilon$. In a particular large region of fixed size, new errors arise on a timescale $O(1/\epsilon)$. If each error dies out---under the unperturbed dynamics---much faster than this timescale, then the errors do not accumulate, and the perturbed steady state resembles the unperturbed one. On the other hand, a large contiguous patch of errors of size $\ell$ takes a time $O(\ell)$ to relax. However, the noise will only generate such configurations with a small probability $O(e^{-\ell})$. These configurations are therefore rare enough that even if they remain far from equilibrium, the state reached at short times will be close to the steady state of the perturbed dynamics. 

We can make this argument rigorous in a way that also generalizes to other deterministic cellular automata. In particular, let us assume that the automaton has the \emph{erosion} property~\cite{ponselet_thesis}, with respect to the steady state $\ket{\vec{0}}$, meaning that in the thermodynamic limit, any initial state with finitely many errors (sites with $b_x=1$) evolves into $\ket{\vec{0}}$ in a finite amount of time\footnote{The restriction to $\ket{\vec{0}}$ is arbitrary, we can also consider erosion w.r.t. to some other steady state, say $\ket{\vec{1}}$. If there are multiple steady states that have the erosion property, they will each be robust to perturbations, at least at the perturbative level. Obviously, the different steady states in this case have to differ from each other at infinitely many locations.}. We will now prove that this, together with the locality of the dynamics, is sufficient to avoid the kind of perturbative instabilities encountered in the previous section and render the steady state stable at any finite order in $\epsilon$. Nevertheless, the possibility of non-perturbative instabilities remains, which is what we will turn to in the next section. 

Let $K_0$ denote the generator of a one-dimensional deterministic CA dynamics and $K_{\epsilon}=K_0 \prod_{x}(1+\epsilon v_{x})$, where
$v_{x}=\sigma_{x}^{+}-\delta_{b_{x},0}$. To get a tight bound, we will assume that $K_0$ exhibits \emph{fast erosion}, by which we mean that any state $\ket{\vec{\Lambda}}$ evolves into the state $\ket{\vec{0}}$ in a time at most $|\vec{\Lambda}|$ (or more generally,  in a time that is linear in $|\vec{\Lambda}|$). Let $b_{x}$ denote the occupation number on site $x$.
In the long-time limit, the probability of $b_x = 1$, starting from the error-free state, can be
written as $
p(\epsilon)=\lim_{t\rightarrow+\infty}\langle b_{x}| K_{\epsilon}^{t}|0\rangle$. Much like we did for Stavskaya, we can expand $p(\epsilon)$ perturbatively as $p(\epsilon)=\sum_n p_n \epsilon^n$ where
 \begin{align}\label{eq:perturbative series}
p_{n} =\sum_{x_{1,\ldots,n}}\sum_{t_{1,\ldots,n}=0}^{\infty}\langle b_{x}|K_0^{t_{n}}v_{x_{n}}\ldots K_0^{t_{1}}v_{x_{1}}|\vec{0}\rangle.
\end{align}
We will now use this expression to argue that $p_n$ has a system-size independent upper bound. This indicates stability at the level of the perturbative expansion. 

To bound $p_n$, we note the following properties of the sum in Eq.~\eqref{eq:perturbative series}:
\begin{itemize}
    \item Every application of $v_{x_{j}}$ on a configuration produces either one or two configurations, while $K_0$ deterministically maps one configuration to another. Therefore we can expand each term in Eq.~\eqref{eq:perturbative series} in terms of some set of spin configurations
    \begin{equation}
    v_{x_{j}}K_0^{t_{j-1}}\ldots K_0^{t_{1}}v_{x_{1}}|\vec{0}\rangle=\sum_{a\in S_{j}}c_{a}\ket{\vec{\Lambda}_a},\label{eq:pertmanyK}
    \end{equation}
    where $\bigl|S_{j}\bigr|\leq2^{j}$. Moreover, it follows directly from the fact that $K_{\epsilon}$ preservers probabilities that $\sum_{a}c_{a}=0$.
    \item The erosion property of $K_0$ implies that the states $\ket{\vec{\Lambda}_j}$ appearing in Eq.~\eqref{eq:pertmanyK} each have at most $j$ errors. Furthermore, the erosion property constrains the possible lists of $t_{j}$ that can appear in Eq.~\eqref{eq:perturbative series}. For example, let us apply $K_0^{t_j}$ to both sides of Eq.~\eqref{eq:pertmanyK}. If $t_j \geq j$, on the RHS we have $\sum_{a\in S_{j}}c_{a}K_0^{t_{j}}|\vec{\Lambda}_{a}\rangle=\left(\sum_{a\in S_{j}}c_{a}\right)|\vec{0}\rangle=0$. In fact, this reasoning can straightforwardly be extended to show that the RHS disappears unless the total
    number of errors added exceeds the number of eroding time steps so we have the constraints 
    \begin{align}\label{eq:TimeConstraint}
    T_{j}\equiv\sum_{k=1}^{j}t_{k}   \leq j, \quad \quad \forall j\in\{1,\ldots,n\}.
    \end{align}
    \item By causality, $b_{x}(t)$ can only depend on the presence or absence of errors at events in the backward lightcone of $(x,t)$. This translates to the constraint 
    \begin{equation}\label{eq:SpaceConstraint}
        |x_{j}-x|\leq \sum^n_{k=j}t_k = T_n-T_{j-1} \leq n,
    \end{equation}
    where in the last step we used Eq.~\eqref{eq:TimeConstraint}.
\end{itemize}

In summary, each term in Eq.~\eqref{eq:perturbative series} is specified by a list of ``noise events'' $\{(x_{j},t_{j})\}_{j=1}^{n}$,  which are constrained by Eqs.~\eqref{eq:TimeConstraint} and~\eqref{eq:SpaceConstraint}. For a fixed such list, the corresponding term in Eq.~\eqref{eq:perturbative series} involves a sum of at most $2^{n}$ configurations. This gives an upper bound
\begin{equation}
|p_{n}|\leq 2^{n}\times(2n)^{n}\times n!=e^{O(n\log n)}.
\end{equation}
Thus $p_{n}$ does not diverge with system size, as it did for the unstable state of Stavskaya, indicating stability at any finite order in perturbation theory. The key step was using the erosion property to argue that the sums over $t_i$ in Eq.~\eqref{eq:perturbative series} can be cut off at a finite time that is independent of $L$. This is a rather general argument, which also applies to CA in higher dimensions and with weaker forms of erosion (i.e., islands eroding at a time $|\vec{\Lambda}|^\alpha$ with $0 < \alpha < 1$).

On the other hand, the bound allows for $p_{n}$ to grow superexponentially with $n$, which would indicate that the series for $p_{n}(\epsilon)$ has zero radius of convergence. Do such non-perturbative instabilities actually occur, and if so, what causes them? We now turn to an example that illustrates this phenomenon.

\subsection{An example of non-perturbative instability}\label{sec:ToomLadder}

We now introduce another deterministic cellular automaton, which we name Toom's ladder (TL) in honor of its originator~\cite{toom_ladder}. It shares with the Stavskaya model the features that it leaves the $\vec{0}$ and $\vec{1}$ states invariant and that has the erosion property\footnote{Note that the erosion property concerns the behavior of states with finitely many errors in the thermodynamic limit. In a finite system, TL actually has many states that do not evolve into either $\ket{\vec{0}}$ or $\ket{\vec{1}}$. Indeed, the existence of such states is related to the model's instability as we will see below.}. Nevertheless, and in contrast to the Stavskaya model, the steady state $\ket{\vec{0}}$ is unstable to arbitrarily weak noise, even when the noise is maximally up-biased. One can find a rigorous proof of this statement in Ref. \cite{toom_ladder}. Here we content ourselves with a more intuitive explanation of the instability which we contrast with the stability seen for Stavskaya.

The model is defined on a ladder geometry, with sites labelled $x=(j,\pm)$. Here $j$ labels the ladder rung, and $\pm$
indicates the top or bottom leg of the ladder respectively. The state space is binary variables $b_{x}=0,1$
on each site, and we will again refer to $1$'s as ``errors''. The deterministic update rules are
\begin{equation}
    b_{j,\pm}\rightarrow \min\bigg(\max(b_{j,+},b_{j,-}),b_{j\mp 1,\pm}\bigg)
\end{equation}

The effects of this rule are simple to describe and verify. First, consider a configuration when all errors are on one leg of the ladder. In that case, the other leg remains error-free and  the errors evolve exactly the same way as they do in Stavskaya, eroding one site each time step from (from the left/right if they are on the upper/lower leg, respectively). For these configurations, TL is therefore equivalent to two mirror image copies of Stavskaya. However, when errors are present on both legs, they interact with each other in non-trivial ways. This can be seen by considering a configuration where all sites on the upper leg are in the $b_x=1$ state. In that case, the errors on the bottom leg do not erode; instead, they simply translate one site to the left in each time step. Similarly when the bottom leg of the ladder is in the $\vec{1}$ state,, errors on the top leg translate to the right. More generally, if the upper leg hosts a large but finite island of errors, then errors ``below'' this island, on the lower leg, can live much longer than they would if the upper leg had no errors. This feature, illustrated in Fig.~\ref{fig:errorisland}, is what will lead to the models instability to noise. 

\begin{figure}
\includegraphics{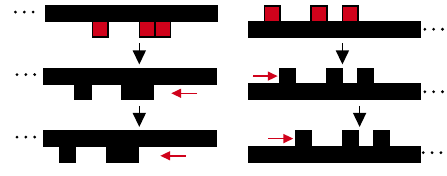}
\caption{This figure demonstrates how a large island of errors on one leg of the Toom Ladder CA (black squares) interacts with new errors introduced in the middle of the island, but on the other leg (red squares). Time runs downwards. The added errors are not eroded, instead they are uniformly translated. This is an important feature of the dynamics which is ultimately responsible for the non-perturbative instability of the TL model.}
\label{fig:errorisland}
\end{figure}

\begin{figure}
    \centering
    \includegraphics[width=1.\columnwidth]{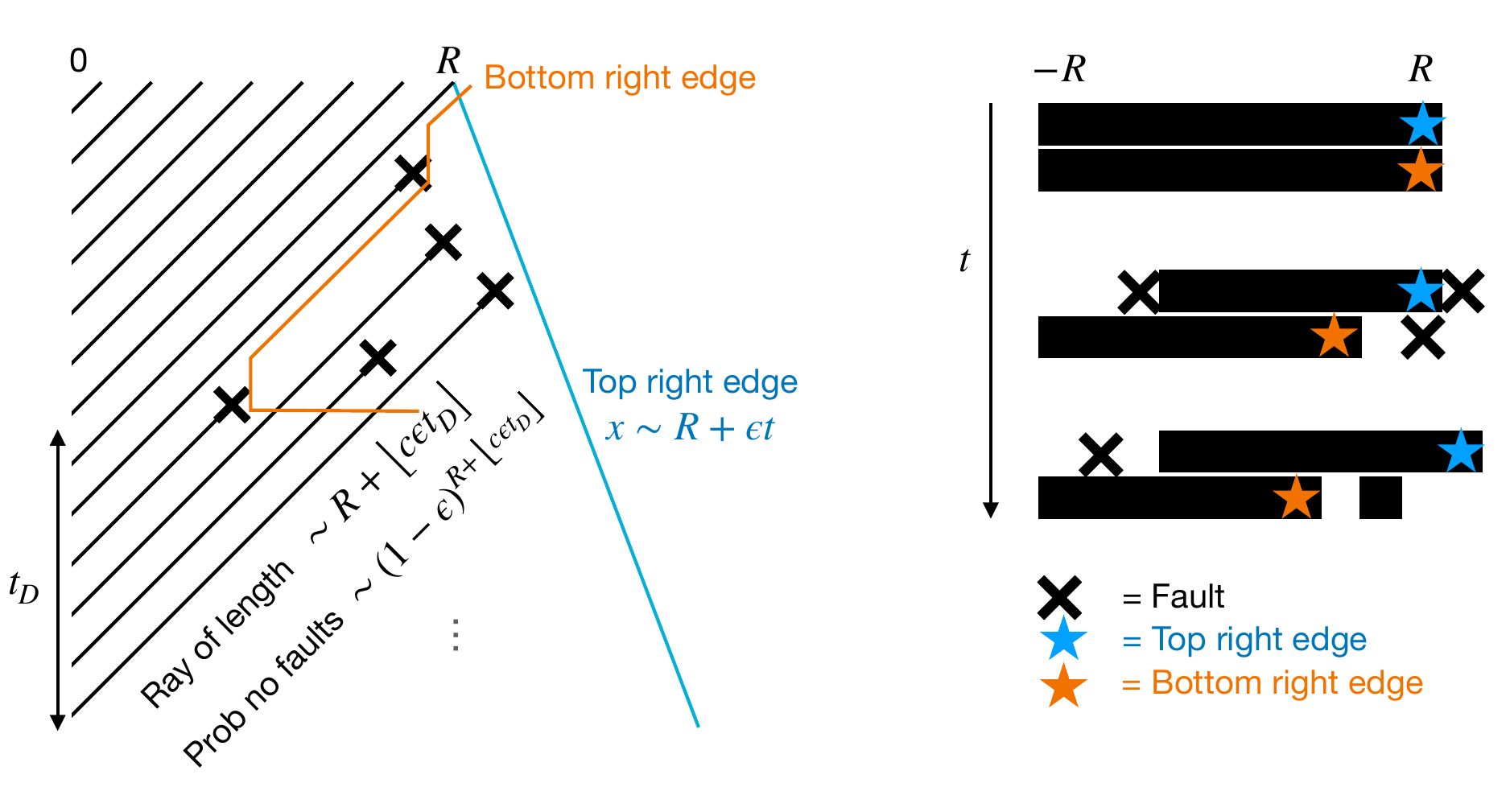}
    \caption{This figure shows the evolution of the right half of an initial error island on both rungs in interval $[-R,R]$ for Toom's ladder. The crosses indicate faults on the bottom rung, the orange line denotes the bottom right end (BRE) of the domain. The blue line denotes the average progress of the top right end of the domain, which undergoes a biased random walk; errors added to the bottom rung between the origin and the blue line are protected from erosion and can thus accumulate, as discussed in the text. An accumulation of errors on the bottom rung implies a finite probability that the BRE never reaches the origin.}
    \label{fig:TL_diagram}
\end{figure}

We now consider what happens when the TL model is perturbed by weak, maximally up-biased noise of strength $\epsilon$. In this case, $\ket{\vec{1}}$ remains a steady state by construction; we will be interested in the fate of the other steady state, $\ket{\vec{0}}$. While the dynamics of the TL model are clearly more complicated than Stavskaya's, it nevertheless satisfies the erosion condition and thus the proof of Sec.~\ref{sec:PertStab} applies to it, indicating that the steady state is stable at any finite order in perturbation theory. Despite this, as mentioned above, the steady state is actually unstable for any $\epsilon > 0$ and it is driven to the state $\ket{\vec{1}}$ by the noise at some time scale that remains finite in the thermodynamic limit (however, the perturbative stability implies that this time scale itself must diverge super-polynomially with $1/\epsilon$ as $\epsilon\to 0$). This means that there must be a non-perturbative mechanism that causes the instability. We now outline this mechanism (referring the reader to Ref. \cite{toom_ladder} for a rigorous proof).

Let us initialize a state which has both rungs of the ladder filled with
errors in the interval $[-R,R]$. Under the unperturbed dynamics,
the top/bottom rung erodes errors from the left/right respectively and the
island splits into two parts, one on the top rung and one on the bottom, at time $t=R$ (see Fig.~\ref{fig:TL_diagram}).
We will show that once $\epsilon>0$, there is a finite probability
the splitting will never occur. In fact, we will argue that the probability of splitting becomes exponentially small in $R$ for large initial islands. 
That implies that large domains, once formed, are very likely to survive
forever. As the up-biased noise favors growing islands of errors, this means that, rather than eroding, distinct islands will eventually start merging into even larger ones, driving the system towards the all $1$ state. 

It suffices to restrict ourselves to the right half of the interval
$[-R,R]$ and to focus on the dynamics of the bottom rung. Specifically,
we track the motion of the rightmost end of the contiguous block of
errors connected to the origin, which we call the bottom right edge
(BRE) of the island (Fig.~\ref{fig:TL_diagram}). We will show that in the presence of noise, the BRE has a chance of never reaching
the origin, and that this event becomes increasingly unlikely with
$R$. 

In the unperturbed dynamics, the BRE will reach the origin
at time $t=R$, which is the time when the splitting of the island into two occurs. We want to know the probability that the BRE is delayed by some time $t_D$ due to additional errors that occur due to the noise acting on the system. The probability that the BRE is delayed by at least one time step is the probability that a fault occurs on the ray connecting
events $(R+1,t=0)$ and $(1,t=R)$, which is given by
\begin{equation}
    \mathbb{P}(t_{D}\geq1)=1-(1-\epsilon)^{R+1}.
\end{equation}
Thus, the probability that no delay occurs goes to zero exponentially as the initial island gets bigger.

For higher values of $t_{D}$ we can argue inductively as follows.
Note that 
\begin{equation}
\mathbb{P}(t_{D}  \geq n+1)=\mathbb{P}(t_{D}\geq n+1\mid t_{D}\geq n)\mathbb{P}(t_{D}\geq n).
\end{equation}
On the other hand, 
\begin{equation}
\mathbb{P}(t_{D}\geq n+1|t_{D}\geq n)=1-\mathbb{P}(t_{D}=n|t_{D}\geq n)
\end{equation}

If we restrict ourselves to the events where $t_{D}\geq n$, then
the event $t_{D}=n$ implies that there can be no faults on the ray
connecting $(R+1,t=n)$ and $(1,t=R+n)$. Therefore 
\begin{align}
\mathbb{P}(t_{D} & =n|t_{D}\geq n)\leq(1-\epsilon)^{R+1}.\label{eq:key_ineq}
\end{align}
so that
\begin{align}
\mathbb{P}(t_{D}\geq n+1) & \geq(1-(1-\epsilon)^{R+1})\mathbb{P}(t_{D}\geq n)\label{eq:refer}
\end{align}
Therefore we have
\begin{equation}
\mathbb{P}(t_{D}\geq n)\geq(1-(1-\epsilon)^{R+1})^{n}
\end{equation}
The right hand side of this inequality goes approximately as 
\begin{equation}
\mathbb{P}(t_{D}\geq n)\gtrsim\exp(-ne^{-\epsilon R}).
\end{equation}
The RHS will be an $O(1)$ quantity, unless $n > e^{\epsilon R}$. Therefore, for any finite $\epsilon$, the lifetime of the island blows up at least exponentially with $R$.

However, an even stronger bound is true because as time progresses,
the top right edge of the system will execute a biased random walk
to the right with speed $O(\epsilon)$ as shown in Fig.~\ref{fig:TL_diagram}. By a time $t$ this will afford
further protection to errors added on the bottom rung with positions
$x>R+O(\epsilon t)$. Thus, the effective value of $R$ used in Eq.~\eqref{eq:key_ineq}
will actually increase with $n$, leading to a tighter bound in Eq.~\eqref{eq:refer}
of the form
\begin{equation}
\mathbb{P}(t_{D}\geq n+1)\geq(1-(1-\epsilon)^{R+\bigl\lfloor c\epsilon n\bigr\rfloor+1})\mathbb{P}(t_{D}\geq n)
\end{equation}
for an unimportant constant $c$. Iterating this, we obtain
\begin{align}
\mathbb{P}(t_{D}\geq n) & \geq\prod_{k=1}^{n}(1-(1-\epsilon)^{R+\bigl\lfloor c\epsilon k\bigr\rfloor+1}).
\end{align}
As $n\to\infty$, the RHS converges to a nonzero number which
can be approximated as $\exp(-e^{-\epsilon R}\frac{(1-\epsilon)^{c\epsilon}}{\epsilon})$.
Therefore, islands have a finite probability of living and growing forever. Inevitably, the system will seed islands which will grow forever and link up to a uniform all $1$ steady state. 

To summarize, we find that the source of the non-perturbative instability is that, while the unperturbed dynamics inevitable erodes errors, this erosion can be delayed forever by even infinitesimally weak noise. The reason for this, and the difference compared to Stavskaya's model where this instability was absent, is that in the TL model there is a much larger ``phase space'' available for adding errors that have the ability to delay the erosion process. In the case of Stavskaya, such relevant errors (which we termed ``faults'' in that case) need to live in the vicinity of the trajectory traced out by the right endpoint of the island under the unperturbed dynamics, a statement that is made precise by the contour argument of Sec.~\ref{sec:contour}. In the TL model, on the other hand, we found that faults that delay the erosion can appear within a much larger region of space-time (see Fig.~\ref{fig:TL_diagram}), one that eventually grows faster than linearly with the size of the island. This increased phase space for faults is what leads to the lack of erosion in the perturbed model and thus to the eventual instability. 

\subsection{A heuristic picture of stability}\label{sec:stability}

The examples of stable and unstable phases considered above suggest two general conjectures, which we formulate in this section. 
First, we argue that a suitably defined erosion property suffices to establish stability in the $\epsilon \to 0$ limit at any fixed order in perturbation theory. 
Second, we argue that the crucial distinction between the Stavskaya and TL models is the entropy of their steady-state spaces: whenever a model obeys a sufficiently strong form of the erosion property, and has only finitely many distinct steady states, we conjecture that the model is stable. This is related to the intuition obtained for the TL model above: it is the number of distinct steady states that controls the ``phase space'' available for adding relevant errors to a space-time history (where ``relevant'' means that they can extend the lifetime of some previous error). The two conditions taken together roughly imply that there is no significant ``leakage'' from the steady state to other very long-lived states; one should compare this to the picture discussed in Sec.~\ref{sec:ASEP} for the ASEP problem where we saw that it is precisely such leakage that can lead to an instability of the steady state. 

\subsubsection{Erosion and perturbative stability}

As we saw in Sec.~\ref{sec:PertStab}, the erosion property of cellular automata is sufficient to render them stable at any finite order of perturbation theory. The key step in the argument was the observation that erosion limits the times $t_i$ that can appear in the perturbative expansion~\eqref{eq:perturbative series}, thus cutting off the time sums at a finite order. Similarly, the locality of the dynamics was used to cut off the sums over the locations of errors $x_i$. 

To generalize the argument, we need to define the analogue of the erosion property to channels that are not deterministic CA. Let us fix some channel $\mathcal{E}$ and one of its steady states $\rho$ in the steady state bundle. We first define an \emph{island of errors} on a region $A$ as a state $\mathcal{V}_A \ket{\rho}$, where $\mathcal{V}_A$ is an arbitrary channel acting on $A$. We will say that the islands has size $R$ if $R$ is the radius of the smallest ball-shaped region that fully contains $A$. We will then say that $\mathcal{E}$ is \emph{$\alpha$-eroding} if 
\begin{equation}
	|| \mathcal{E}^t \mathcal{V}_A  \ket{\rho} - \ket{\rho}||_1 \leq e^{-O(t^\alpha / R)}
\end{equation}
holds for any island of size $R$. 

We claim that if $\mathcal{E}$ is $\alpha$-eroding with $\alpha > 0$ then it has perturbative stability. For simplicity, we again consider on-site perturbations of the form $\prod_{x}(1+\epsilon v_{x})$. We want to bound the terms that appear at $n$-th order in perturbation theory for the expectation value of some local observable\footnote{While we focus on local observables here, the same argument should go through for the kind of non-local order parameters one might use to diagnose, say, topological order (i.e. the behavior of Wilson-line operators), provided one takes the thermodynamic limit first.} $o_x$ localized near position $x$. These once again take the form
\begin{equation}\label{eq:PertSeriesAgain}
    \sum_{x_{1,\ldots,n}}\sum_{t_{1,\ldots,n}}^{t}\braket{o_x|\mathcal{E}^{t_{n}}v_{x_{n}}\ldots\mathcal{E}^{t_{1}}v_{x_{1}}|\rho}.
\end{equation}
Let $T_j \equiv \sum_{k=1}^{j} t_k$ denote the total time elapsed before the $j$-th error event. 

\begin{figure}
    \centering
    \includegraphics[width=1.\columnwidth]{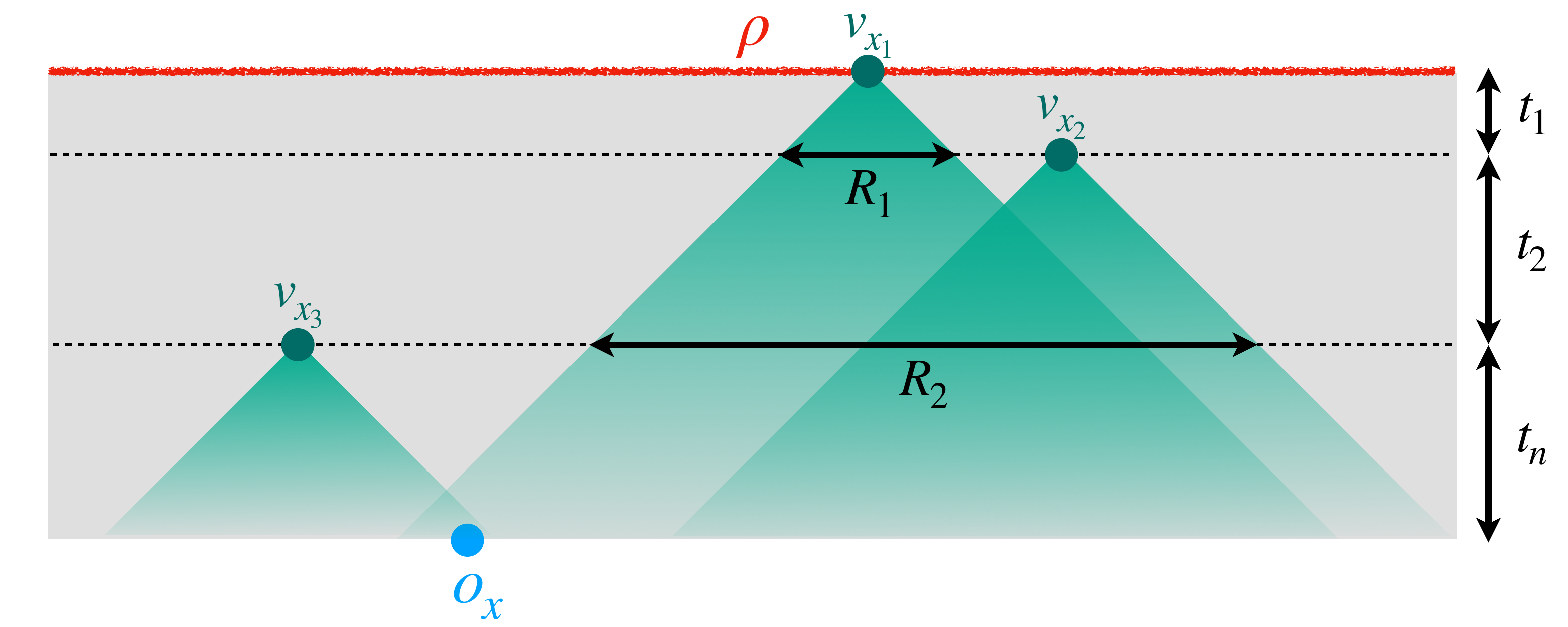}
    \caption{Sketch of general erosion argument. The errors $v_{x_i}$ in Eq.~\eqref{eq:PertSeriesAgain} create small islands on top of the steady state $\rho$. These islands spread out within the Lieb-Robinson cone but they also get continuously eroded and would disappear unless they encounter other another error that can extend their lifetime. They all need to link up with each other and the observable $o_x$ to give a non-vanishing contribution.}
    \label{fig:erosion}
\end{figure}

The argument for stability runs analogous to the one in Sec.~\ref{sec:PertStab}. let us first consider the state $\mathcal{E}^{t_1} v_{x_1} \ket{\rho}$: due to the erosion this will be very close to $\ket{\rho}$ after some $O(1)$ timescale $t_1^*$ and we can cut off the sum over $t_1$ at this time. On the other hand, due to the Lieb-Robinson bound, the state $\mathcal{E}^{t_1^*} v_{x_1} \ket{\rho}$ can be well approximated by an island of size $R_1 \approx 2vt_1^*$ (where $v$ is the L-R velocity). Under subsequent evolution this island will erode, unless the next error, $v_{x_2}$ influences it in a way that extends its lifetime. In order to do so, and again using Lieb-Robinson, this new error has to be close enough, which we can loosely bound by requiring that the state $v_{x_2} \mathcal{E}^{t_1} v_{x_1} \ket{\rho}$ can be approximated as a single island of size at most $\tilde{R}_1 = O(R_1 + R_1^{1/\alpha})$. Now, we can apply the erosion condition to this larger island, which implies that the sum on $t_2$ can be cut off at a time $t_2^* = O(\tilde{R}_1^{1/\alpha})$ which is clearly an $O(1)$ quantity independent of system size. 

We can now iterate this argument for all subsequent times. Since we cut off $t_2$, $\mathcal{E}^{t_2}v_{x_2} \mathcal{E}^{t_1} v_{x_1} \ket{\rho}$ can be approximated by an island of size at most $R_2 = \tilde{R_1} + 2 v t_2^*$. In order to keep the previous errors from eroding, $v_{x_3}$ has to be sufficiently close to this island, which means that the state $v_{x_3} \mathcal{E}^{t_2}v_{x_2} \mathcal{E}^{t_1} v_{x_1} \ket{\rho}$ itself is an island of size at most $\tilde{R}_2 = R_2 + R_2^{1/\alpha}$, putting a finite cutoff $t_3^* = O(\tilde{R}_2^{1/\alpha})$ on $t_3$. We can iterate this all the way up to $t_n$ to ensure that all the time sums can be cut off at finite, system size independent values. Finally, we can use the LR bound once more to argue that error far from $x$, with $|x_j -x| > \sum_{k=j}^{n} t_k = T_n - T_{j-1}$ have exponentially suppressed contributions. This allows us to also cut off the sums over $x$ at finite values, thus rendering the whole expression in Eq.~\eqref{eq:PertSeriesAgain} finite. 

\subsubsection{Conditions for nonperturbative stability}

What about stability beyond perturbation theory? Here we will have to resort to conjectures, and we will use the lessons learned from Toom's ladder to motivate them. 

One way that stability can fail is by the erosion being too slow. While for the argument above, we could get away with any $\alpha > 0$, this might not be enough to preclude a non-perturbative instability. In particular, if we add noise that grows an island at its boundary with some constant rate, this will overpower the erosion of the unperturbed dynamics immediately, unless $\alpha=1$. We will refer to models that have $\alpha=1$ erosion as \emph{fast eroders}\footnote{Nevertheless, $\alpha <1$ erosion might still be sufficient for stability if we cannot add noise that deterministically grows islands of errors. This could happen due to symmetry restriction, as in 2D Glauber dynamics when spin-flip symmetry is enforced, or due to the fact that the different steady states are locally indistinguishable, as in the (four-dimensional) toric code.}. 

Even fast erosion is not sufficient to rule out instabilities, as the TL model illustrates. In our analysis of that model, we also saw the underlying reason for the eventual instability: the fact that there are many different islands of the same size meant that there are many different ways of adding errors to a space-time history in a way that increases the time before the island can erode.  We conjecture that, at least for fast eroders, this is the only way to get an instability. However, to formulate it for generic models, we first need to introduce some definitions.

First, we want to define a \emph{contiguous island of errors}. We will call an island $\rho_A$ contiguous if for any observable $o$, we have that $\frac{\mathrm{d}}{\mathrm{d}t} \braket{o|\rho_A} \leq O(e^{-\text{dist}(\text{supp}(o),\partial A)})$, where $\text{dist}(\text{supp}(o),\partial A)$ is the distance between the region on which $o$ is supported and the boundary of $A$ (we also assume that the region $A$ on which the island is supported is itself contiguous in space). This means that $\rho_A$ looks like a steady state away from $\partial A$ and, by Lieb-Robinson, it ensures that observables supported more than $vt$ away from the boundary will remain static up to time $t$. We expect that such contiguous islands will be the slowest to erode. 

Next, we want to count the number of \emph{distinct} contiguous islands. We say that two contiguous islands $\rho_A^{(1,2)}$ are distinct if $||\rho_A^{(1)} - \rho_A^{(2)}||_1 \geq C$ where $C$ is some $O(1)$ constant. We expect that the number of contiguous islands on $A$ that are all distinct from one another will be a finite, countable number, which we will denote by $N(A)$, and let $N(R)$ be the maximum of $N(A)$ over all islands of size $R$. Thinking back to our CA examples, in Stavskaya, we have $N(R) = 1$ for any $R$: the only contiguous island is an insertion of the other steady state $\ket{\vec{1}}$. In TL, on the other hand, there are many different contiguous islands, where one leg hosts a block of errors but the other can have different steady state configurations in its bulk. This means that $N(R)$ is some growing function of $R$. As we have seen the nonperturbative instability of TL stems from the fact that the noise can introduce perturbations in any temporal order and have them be long-lived: this is precisely because these perturbed configurations are themselves contiguous islands. Heuristically, absent a finite entropy of contiguous islands, there are only finitely many places (or, more generally, a number of places proportional to the boundary $|\partial A|$) where one can add noise to a contiguous island in a way that makes the island live longer.
Motivated by this observation, we formulate the following conjecture:

\begin{conjecture}\label{conj:stability}
    Let $\mathcal{E}$ be a channel in $d$ spatial dimensions that has fast erosion w.r.t to its steady state $\rho$. Furthermore, assume that relative to the steady state $\rho$, the number of distinct islands satisfies $N(R) \leq O(R^{d-1})$. Then $\rho$ is stable to arbitrary local perturbations of $\mathcal{E}$. 
\end{conjecture}
For example, the conditions in the conjecture should apply to the 2D Toom cellular automaton, where there is an $O(R)$ number of islands. 

We expect $N(R)$ to be related to the number of (almost) steady states on a finite system. For example, the TL model on a ladder of length $R$ with periodic boundaries has $e^{O(R)}$ different eigenvalues that have modulus one, and the corresponding eigenstates are closely related to the configurations that make the lifetime of islands infinite in the perturbed model (see our discussion in Sec.~\ref{sec:ToomLadder} above). More generally, we expect that various distinct islands can be obtained from eigenstates of the dynamics restricted to $A$ (with some appropriate choice of boundary condition) that have eigenvalues that approach $1$ in the limit of large $R$. Thus we conjecture that unstable fast eroders have a large number of steady states in the thermodynamic limit\footnote{More generally, we should also consider eigenstates whose eigenvalues are getting close to some point on the unit circle, not necessarily at $1$. We expect that when there is sufficiently many of these, we should still be able to use them to construct contiguous islands that appear stationary for times longer than it takes for them to erode.}. As a consistency check, in App.~\ref{app:soldier} we consider two more examples of 1D CA models that are known to be fast eroders but are unstable to noise. We confirm that both indeed have a number of steady states that grows with system size. 

Here, we formulated what we conjecture to be a sufficient condition for stability in the sense that local expectation values should evolve smoothly with the strength of perturbations. As discussed in Sec.~\ref{sec:uniformity_implies}, this sort of stability is a natural consequence of uniformity as we have defined it. To round up our set of conjectures, we therefore further conjecture that if a channel $\mathcal{E}$ and some set $\mathcal{S}$ of its steady states all satisfy the conditions of Conjecture~\ref{conj:stability}, then there exist a uniform steady state bundle, with finite volume in the space of all channels, that contains $(\mathcal{E},\mathcal{S})$.

\section{Conclusion}\label{discussion}
We proposed a condition called ``uniformity'', which guarantees that a family of open many-body systems, classical or quantum, has stable steady state correlations. Briefly stated, uniformity within some open set of parameter space requires that any point in this set is surrounded by a neighborhood $N$ such that (i)~the steady states at any two points in $N$ can be placed in a natural one-to-one correspondence, and (ii)~the steady state for any point $x \in N$, evolved under the dynamics at any other point $x' \in N$, relaxes fast to the corresponding steady state at $x'$. The physical picture underlying this criterion is that steady states in the same phase differ only by the presence of "errors" that can be healed locally and rapidly, even if the states themselves are exotic and harbour long-ranged correlations. On the other hand, steady states in distinct phases (e.g., a product and long-range correlated state) will fail to obey (ii), as product states cannot rapidly develop long-ranged correlations (they can only do so slowly, over a time that scales with system size).

We also explored the connection between uniformity and spectral properties of the corresponding quantum channels generating the dynamics. We showed that if a channel $\mathcal{E}$ lies in a uniform region, then there is a system-size independent upper bound on the growth of  $\braket{o|\left((\mathcal{I}-\mathcal{E})^{-1}\mathcal{V}\right)^{n}|\omega} < e^{O(n)}$, where $o$ is a local observable, $\omega$ the steady state, and $\mathcal{V}$ is a local superoperator perturbing the channel. The exponential bound on this matrix element resembles a gap condition on the resolvent, although the bound is \emph{not} implied by the existence of a gap (as the ASEP example of Sec.~\ref{sec:ASEP} shows). What role does a spectral gap play in the story of uniformity? In our one-dimensional Stavskaya example (in the non-trivial bistable phase) we were able to show that the model obeys uniformity. It also has a spectral gap, and we numerically demonstrated that this gap coincides with the exponential decay rate that appears in the definition of uniformity, confirming the intuition gained from the ASEP model. We have conjectured that this relation between gap and the decay rate holds more generally for the exponentially decaying version of uniformity. We also noted that a weaker version of uniformity (with sub-exponential, but sufficiently fast decay) is powerful enough to imply the stability of correlations, and might hold even in the absence of a spectral gap.

Finally, we looked for a sufficient condition on a channel which would guarantee that it lies in a uniform region of parameter space, and hence defines an open phase of matter. We studied known examples of stable and unstable one-dimensional cellular automata, and  identified that cellular automata can have perturbative and non-perturbative instabilities. We formulated a sufficient condition (erosion) which prevents the former, and conjectured a stronger condition which prevents the latter (Conjecture \ref{conj:stability}).

\subsection{Outlook and further discussion}
Some important comments about the range of validity of our results and possible extensions are in order.

\paragraph*{Quantum open systems and topological order:} While our discussion has focused on classical phases, it also has implications for stable quantum phases, such as the four-dimensional toric code, that are capable of passive quantum error correction. Assuming the 4D toric code (with some local error correction, e.g. \cite{breuckmann2017local,kubica2019cellular}) lies in a uniform region of parameter space (as we conjecture), what are the implications? Theorem~\ref{thm:isomorphism} implies that each of the steady states of a general channel $\mathcal{E}'$ in the uniform region can be put into 1:1 correspondence with steady states of the fixed point toric code model ($\mathcal{E}$) through a finite chain of low-depth channels, and vice versa. So if $\rho_a$ is a reference logical qubit at the fixed point, it relaxes quickly to a steady state $\rho_a'$ of $\mathcal{E}'$ under dynamics  $\mathcal{E}'$. We also know that $\rho_a'$ relaxes quickly to \emph{some} steady state of $\mathcal{E}$ under dynamics $\mathcal{E}$. In the present case, however,  the  diverging code distance between topological sectors and the low depth of the channels implies that $\rho_a'$ relaxes quickly back to the particular state $\rho_a$! Therefore the channels $\mathcal{E}^t, (\mathcal{E}')^t$ are approximate inverses of one another when restricted to the steady state space and at sufficiently large times set by $\Delta$. However, a channel  only ever has a channel as its inverse when it is unitary, suggesting that the steady state spaces of the two channels (and the corresponding logical operators) are approximately related by a low-depth \emph{unitary} channel. Therefore, in this case, the steady states are related in a similar manner to ground states of the perturbed toric code model, and encode the same logical qubits. Note that even without access to this unitary channel, one can perform logical operations on a steady state $\rho'$ of $\mathcal{E}'$ as follows: Use $\mathcal{E}$ to evolve $\rho'$ to a fixed point state, perform the desired Pauli logical operation, and then use $\mathcal{E}'$ to evolve the resulting state back to the original steady state space.

\paragraph*{Is uniformity necessary for stability?} 
Uniformity is a sufficient condition for the stability of steady-states, but is it necessary too? As noted in App.~\ref{app:uniformity}, the uniformity definition can be weakened to encompass cases in which the perturbed channel approaches its steady state as a power law rather than an exponential. If the power law is sufficiently rapid, the corollaries of uniformity (LPPL, stability of correlations) go through as in the exponentially decaying case. Thus we expect a generalised but still powerful "fast power-law" version of uniformity to hold for a wider class of models that may include gapless stable phases (e.g., Goldstone phases in high dimensions). 

Is the generalised "fast power-law" version of uniformity necessary for phase stability?  If nearby channels relax to each other's steady states as a slow enough power law, our arguments---in particular for LPPL---break down. However, the failure of these arguments for slower decay might not entail the instability of the phase. A simple counterexample is a particle diffusing on a lattice of sites with random on-site energies; the hopping rates between neighboring sites are set by detailed balance at some equilibrium temperature. It is straightforward to show that the effects of locally perturbing the energy landscape die out diffusively, as $t^{-1/2}$, at any finite distance. This decay is too slow for our proof of LPPL to work. Nevertheless, the effect of local perturbations on steady-state populations is strictly local in the thermodynamic limit. More generally, the constraints imposed by conservation laws can cause phases to be stable even when our counting arguments suggest they are not. Exploring the effects of these constraints is an interesting topic for future work.

\paragraph*{Further open questions:}We conclude by listing some open questions raised by our work. First, we listed a number of conjectures throughout the paper (\ref{conj:uniformity=gap}, \ref{conj:stability}) and proving or refuting any of them would clarify our picture of open phases in important ways. Establishing uniformity in cases beyond our 1D example, in particular for Toom's rule or four-dimensional toric code, is another important issue. Also, it would be useful to prove or refute our conjecture (Sec.~\ref{sec:compare_gap_uniformity}) that uniformity (in a volume) implies that the corresponding models are in the same \emph{dynamical} universality class. Finally, there are various directions in which our definition of open phases could be generalized. While we assumed time-translation invariance, this could probably be easily relaxed to more general, time-dependent evolutions. A more challenging direction is to extend these ideas to non-Markovian dynamics, which is relevant to many experimental settings. 

\begin{acknowledgments}
The authors thank David Huse, Marko \v{Z}nidaric,  Wojciech de Roeck, Fiona Burnell, and Chaitanya Murthy for insightful discussions. S.G. acknowledges support from NSF DMR-2236517. T.R. is supported in part by the Stanford Q-Farm Bloch Postdoctoral Fellowship in Quantum Science and Engineering. CvK is supported by a UKRI Future Leaders Fellowship MR/T040947/1.
\end{acknowledgments}

\appendix

\section{Thermal stability and the 2D Ising model}\label{App:thermal_2D_Ising}

In this Appendix, we discuss some known results for the stability of finite-temperature quantum phases and prove some of our own, in particular showing that the LPPL propety is satisfied under appropriate conditions. 

To begin, we note that Ref. \cite{Hastings_belief} provides an idea that bears resemblance to the quasi-adiabatic continuation of ground states discussed in the main text. For a 1-parameter family of Hamiltonians $H_{s}=H+Vs$, where $V=\sum_{x}v_{x}$ is a sum of local terms, it is shown that 
\begin{equation}
\partial_{s}e^{-\beta H_{s}}  =\eta_{s}e^{-\beta H_{s}}+e^{-\beta H_{s}}\eta_{s}^{\dagger},\label{eq:almost_quasi}
\end{equation}
where
\begin{align}
\eta_{s} & =-\frac{\beta}{2}\delta V+\mathrm{i}K_{s}\nonumber \\\delta V_{s} & \equiv V-\bigl\langle V\bigr\rangle\label{eq:delta_V}\\K_{s} & \equiv\sum\tilde{v}_{x,s}\label{eq:K_def}\\\tilde{v}_{x,s} & \equiv-\sum_{n\geq1}\int_{-\infty}^{\infty}dt\:\mathrm{sign}(t)e^{-2\pi n|t|/\beta}e^{iH_{s}t}v_{x}e^{-iH_{s}t}\end{align}$\delta V_{s}$ is the difference between $V$ and its expectation value in the instantaneous density matrix (this subtraction ensures that $\mathrm{tr}(\rho_{s})=1$). $\mathrm{i}K_{s}$ is anti-Hermitian, and it follows from applying Lieb-Robinson bounds in the integral that it is a sum of bounded terms $\tilde{v}_{x}$ which are exponentially localized with spatial extent $O(\beta)$. While $\eta_{s}$ is a sum of local terms, we would also need it to be anti-Hermitian in order to integrate up Eq.~\eqref{eq:almost_quasi} and show that $e^{-\beta H_{s}}$ is related to $e^{-\beta H_{0}}$ by a local unitary transformation. However, this is generically not true due to the $\delta V$ term in Eq.~\eqref{eq:almost_quasi}. Thus, a direct analogy with quasi-adiabatic continuation fails. Therefore, unsurprisingly, one cannot argue that the thermal density matrices at different $s$ have qualitatively similar correlation functions without making additional assumptions.

We now formulate some such additional assumptions. We will show that the expectation values of local observables change differentiably with $s$ if we assume that the $v_{x},\tilde{v}_{x}$ have suitably short-ranged correlations with all said observables. Eq.~\eqref{eq:almost_quasi} implies that the response of local unit norm observable $o$ (with support $d_{o}$, located near the origin $0$) to an infinitesimal change in $s$ is
\begin{align*}|\partial_{s}\langle o\rangle| & \leq\frac{\beta}{2}\bigl|\left\langle \{o,\delta V\}\right\rangle \bigr|+\bigl|\left\langle \left[o,\mathrm{i}K\right]\right\rangle \bigr|.\end{align*}
We have kept the $s$ dependence of states and operators implicit for brevity. The second term is a commutator between local observable $o$ and a sum of exponentially local and bounded terms; consequently it is uniformly bounded in the thermodynamic limit. The first term can be recast as sum of connected correlators $\sum_{r}\bigl\langle\{\delta o,\delta v_{r}\}\bigr\rangle$. This term is uniformly bounded in the thermodynamic limit provided $v,o$ have sufficiently rapidly decaying connected correlations in the thermal state (it suffices for the correlation to decay faster than $Cr^{-D}$ where $D$ is the number of spatial dimensions). Supposing this holds, we can extract a bound with the following the parametric dependencies on $\beta,d_{o}$\begin{align}|\partial_{s}\langle o\rangle| & \leq\beta C'\mathrm{poly}(d_{o})\bigl\Vert v\bigr\Vert+C''\beta^{D+1}\bigl\Vert v\bigr\Vert d_{o}\label{eq:local_op_bound}\end{align}
These bounds can be sharpened if assumptions are made about the geometry of the support of $o$ (e.g., if $o$ acts on a co-dimension $>0$ surface) but that will not be necessary for our present purposes. What matters is that the RHS of Eq.~\eqref{eq:local_op_bound} is finite in the thermodynamic limit. Therefore, the expectation values of local observables in a thermal state $e^{-\beta H_{s}}$ have a bounded derivative with respect to $s$ provided first $\beta<\infty$ and secondly
\begin{align}\bigl\langle\{\delta o,\delta v_{r}\}\bigr\rangle & <O(|r|^{-D})\label{eq:single_obs}\end{align}
for all local $o$. If these conditions hold uniformly for some interval, say, $s\in[0,1]$ then it follows that $\langle o\rangle_{s}$ is continuous and differentiable on this interval, which excludes the possibility of a first or second order phase transition for these values of $s$. Colloquially, this result says that thermal phases are stable to perturbations that have short-ranged correlations with all local observables. 

As an example, consider the 2D quantum Ising model at $0<T<T_{c}$. This is a symmetry-protected phase, insofar as the long-range order is expected to be stable in a finite region of parameter space, provided the perturbations preserve the protecting Ising symmetry. This is consistent with our above observation. Preserving the symmetry amounts to choosing $v_{r}$ (and consequently $\tilde{v}_{r}$) which are Ising symmetric. Ising symmetric observables, on the other hand, have only short-ranged connected correlations in the symmetry broken phase, so that Eq.~\eqref{eq:single_obs} holds.

Note that even Ising symmetric perturbations will destroy the phase once they are made large enough. For example, take $V=-H$ such that increasing $s$ effectively increases the temperature. At sufficiently large $s$ the ordered phase evaporates. At this phase transition, connected correlation functions must decay slowly enough to invalidate our bound Eq.~\eqref{eq:single_obs}. For example, the energy density will have a divergent derivative at the transition. This is consistent with a breakdown of Eq.~\eqref{eq:single_obs}, because at the critical point the energy-energy correlation function does not quite decay fast enough ($\bigl\langle\{\delta h_{0},\delta h_{r}\}\bigr\rangle\sim r^{-2}$). It is helpful to extend this result to consider the response of two-point connected correlation functions $C_{rr'}\equiv\bigl\langle o_{r}^{\phantom{'}}o'_{r'}\bigr\rangle-\bigl\langle o_{r}\bigr\rangle\bigl\langle o'_{r'}\bigr\rangle$ to an infinitesimal change in $s$. Eq.~\eqref{eq:almost_quasi} gives that
\begin{align}
\partial_{s}C_{rr'}&=-\frac{\beta}{2}\bigl\langle\{\delta V,\delta o_{r}\delta o'_{r'}\}\bigr\rangle \nonumber\\ 
                            &+\bigl\langle\delta([o_{r},\mathrm{i}K])\delta o'_{r'}\bigr\rangle+\bigl\langle\delta o_{r}\delta([o'_{r'},\mathrm{i}K])\bigr\rangle.
 \label{eq:partial_C}\end{align}
All three terms are connected correlations. Assume $o,o'$ have unit operator norm. The first term  consists of a sum of terms of form $\bigl\langle\{\delta v_{x},\delta o_{r}\delta o'_{r'}\}\bigr\rangle$. The second and third terms are made up of expressions like $\bigl\langle\delta[o_{r},\tilde{v}_{x}]\delta o'_{r'}\bigr\rangle,\bigl\langle\delta[o'_{r'},\tilde{v}_{x}]\delta o_{r}\bigr\rangle$. Provided these all decay fast enough
\begin{align}
&\bigl\langle\{\delta v_{x},\delta o_{r}\delta o'_{r'}\}\bigr\rangle,\bigl\langle[o_{r},\tilde{v}_{x}]o'_{r'}\bigr\rangle,\bigl\langle[o'_{r'},\tilde{v}_{x}]o_{r}\bigr\rangle\nonumber\\
&<O(\mathrm{min(}|r-x|^{-D},|r'-x|^{-D}))\label{eq:ccbound},\end{align}
we can bound the RHS of Eq.~\eqref{eq:partial_C} uniformly both in system size and in $r,r'$. Therefore the two point correlation functions in a thermal state $e^{-\beta H_{s}}$ have a bounded derivative with respect to $s$ provided $\beta<\infty$ and Eq.~\eqref{eq:ccbound} holds. If these conditions hold uniformly for some interval, say, $s\in[0,1]$ then it follows that $C_{rr'}$ is continuous and differentiable on that interval. 

Suppose that  $o,o'$ exhibit long range order at $s=0$, i.e., $C_{rr'}\rightarrow O(1)$ as $r\rightarrow\infty$. If one can show that $o,o'$ have short-ranged correlations with $v,\tilde{v}$ over a range of $s$, then the result above implies that their long range order will survive for a non-vanishing range in $s$. For example, in the 2D quantum Ising model for $T<T_{c}$, the local order parameter exhibits long-ranged order i.e., $C_{rr'}=\left\langle \sigma_{r}^{z}\sigma_{r'}^{z}\right\rangle _{c}\rightarrow O(1)$ as $r\rightarrow\infty$. Indeed this defines the Ising ordered phase. The stability of these long-ranged correlations to Ising symmetric perturbations is consistent with our analytical result. That is because such perturbations require Ising symmetric $v,\tilde{v}$. Connected correlations involving Ising symmetric variables are expected to be (exponentially) local, so that Eq.~\eqref{eq:ccbound} is easily satisfied throughout the phase. This again demonstrates the rather intuitive idea that phases are stable to perturbations that have short-ranged correlations with the observables defining the phase.

\section{Refined definition of uniformity}\label{app:uniformity}
The main text presents uniformity as the requirement that neighbouring
states in the steady-state bundle relax to one another ``exponentially
quickly'' $\bigl\Vert\mathcal{E}^{t}(\rho')-\rho\bigr\Vert<e^{-O(\Delta t)}$.
However this mathematical relationship is plainly ambiguous insofar
as a norm was not specified. This appendix presents a (more) precise
statement of uniformity; it is the requirement that neighbouring states
relax quickly to one another according to the expectation values of
local observables. This is captured by the mathematical statement
that 
\begin{equation}
    |(\mathcal{E}^{t}(\rho')-\rho)(o)|  \leq C\min(1,f(t,d)),\label{eq:rapid_mix_new}
\end{equation}
for all operators $o$ with finite support and supremum
norm $\bigl\Vert o\bigr\Vert_{\infty}=1$. Here $d=d_{o}$ denotes the support of operator $o$; depending on context, that could mean either the number of sites on which $o$ has non-trivial support (the `weight'), or the diameter of a minimal ball containing $o$ (the `diameter'). We will clarify which notion of support we are referring to when it is important. $C$ is a constant independent of $t,d$.

Eq.~\eqref{eq:rapid_mix_new} is more precise and also more general than our initial exponential bound defining uniformity, and it allows for the possibility that variables with larger support relax more slowly. This appendix provides some sufficient conditions on $f$ such that Eq.~\eqref{eq:rapid_mix_new} implies analyticity, LPPL, and the stability of long-ranged correlations. In all cases $f(t,d)\geq0$ will be some function which non-increasing with $t$, and non-decreasing with $d$. Moreover, $f$ does eventually decay to zero, so there certainly exists a time $t_{d}$ such that $f(t,d)\leq1$ when $t\geq t_{d}$.

When uniformity is reformulated using Eq.~\eqref{eq:rapid_mix_new}
in place of simple exponential decay, we find that it implies a form
of LPPL if $f(t,d)$ decays faster than $1/t$. We also find that
an even weaker condition, namely $\lim_{t\rightarrow\infty}f(t,d)=0$
which has already been assumed, implies analyticity of local observables. In these two cases it will not matter how $f$ depends on $d$ or whether $d$ refers to the weight or diameter. This turns out to be the case because analyticity and LPPL, as we have formulated them, refer to properties of 1-point functions of local observables. It is arguably surprising that the proofs go through under such mild
conditions. Indeed, we believe that while the slow decay laws suffice
to show LPPL and analyticity, the stability of laws of the form Eq.~
\eqref{eq:rapid_mix_new} themselves require a faster decay of $f$
with $t$. To that end, we provide a non-rigorous argument that stable
phases have $f(t,d)$ that decays at least as fast as $t^{-D-1}$,
where $D$ is the spatial dimension. Finally, to show the stability of long-ranged correlations we assumed that $f(t,d)$ decays in time. The only other constraint is that $f(t,d)$ is allowed to grow in any way with the weight of $o$, but it is not allowed to grow with the diameter of $o$. 

Note that in the Stavskaya phase one can rigorously bound relaxation with the function $f(t,d)= d e^{-\Delta t}$ where $d$ refers to operator weight (Sec.~\ref{sec:stavskayauniform}). It follows from the above summary that this form of $f$ is sufficient to prove analyticity, LPPL, and the stability of long-ranged order. 

\subsection{Analyticity}

We examine a family of channels $\mathcal{E}_{\epsilon}$, $\epsilon\in[-1,1]$
entirely contained within one of the balls used to define uniformity.
Our goal is to argue that local correlation functions are analytic
in $\epsilon$. Let $\omega_{0}$ be any phase steady state of $\mathcal{E}_{0}$,
and let $\omega_{\epsilon}$ be the corresponding steady state of
$\mathcal{E}_{\epsilon}$ which $\omega_{0}$ rapidly relaxes to.
Then uniformity gives

\begin{equation}
|\bigl\langle o\bigr\rangle_{\epsilon}-\bigl\langle o\mid\mathcal{E}_{\epsilon}^{t}\omega_{0}\bigr\rangle|\leq C\min(1,f(t,d))\label{eq:analytic_1}
\end{equation}

where $\bigl\langle o\bigr\rangle_{\epsilon}\equiv\bigl\langle o\mid\omega_{\epsilon}\bigr\rangle$.
Provided $f$ tends to zero at large $t$, $o(\epsilon;t)\equiv\bigl\langle o\mid\mathcal{E}_{\epsilon}^{t}\omega_{0}\bigr\rangle$
is a sequence of bounded analytic functions in $\epsilon$ which tend
to $\bigl\langle o\bigr\rangle_{\epsilon}$ \emph{uniformly}
as $t\rightarrow\infty$. Uniformity comes from the fact that the
RHS of Eq.~\eqref{eq:analytic_1} is independent of $\epsilon$. Analyticity
comes from the fact that $o(\epsilon;t)$ can be expressed as a finite
time evolution on a finite quantum system (using Lieb-Robinson bounds).
Therefore uniformity (with the condition $\lim_{t\rightarrow\infty}f(t,d)=0$)
implies the analyticity of local observables. 

\subsection{LPPL}
Our starting point is Eq.~\eqref{eq:lppl_master} in the main text. Using this with uniformity, with the new bound Eq.~
\eqref{eq:rapid_mix_new}, yields an inequality

\begin{equation}
\bigl|\langle o_{x}\mid\rho-\rho'\rangle\bigr|\leq C\min(1,f(t,d))+\sum_{s=1}^{t}\bigl|\lozenge_{s}\bigr|.\label{eq:lppl_reminder}
\end{equation}

The function $f$ must obey certain conditions if our proof of LPPL is to go through
as before. First we must bound the large $s$ contributions from the
second term. These are controlled by
\begin{equation}
F(t,d)  \equiv\sum_{s>t}f(s,d).
\end{equation}
The function $f$ must decay fast enough that this sum exists. In that case $F$
is non-increasing with $t$ and eventually decays to zero. We can
therefore define a function $t_{\epsilon,d}$ such that $F(t,d)<\epsilon$
whenever $t\geq t_{\epsilon,d}$. $t_{\epsilon,d}$ can be chosen
to increase monotonically with $\epsilon$ and we can also assume
$t_{\epsilon,d}>t_{d}$. Note that $t\geq t_{\epsilon,d}$ implies
$f(t,d)<\epsilon$ too. Taking $t>t_{\epsilon,d}$ in Eq.~\eqref{eq:lppl_reminder}
gives
\begin{equation}
\bigl|\langle o_{x}\mid\rho-\rho'\rangle\bigr| \leq Cf(t,d)+\sum_{s=0}^{t_{\epsilon,d}}\bigl|\lozenge_{s}\bigr|+\sum_{s=t_{\epsilon,d}+1}^{t-1}\bigl|\lozenge_{s}\bigr|
\end{equation}
where we have split the time sum into $s\leq t_{\epsilon,d}$ and
$s>t_{\epsilon,d}$ in the second and third terms respectively. Using
the fact that
\begin{equation}
\left|\lozenge_{s}\right|=\left|\langle o_{x}\mid\mathcal{E}{}^{s+1}-\mathcal{E}{}^{s}\mid\rho'-\rho\rangle\right|\leq2C\min(1,f(s,d))
\end{equation}
the last term may be bounded using uniformity 
\begin{equation}
\sum_{s=t_{\epsilon,d}}^{t}\bigl|\lozenge_{s}\bigr|  \leq2CF(t_{\epsilon,d},d)\leq2C\epsilon.
\end{equation}
The first term is also bounded by $\epsilon$, so we have
\begin{equation}
|\langle o_{x}\mid\rho-\rho'\rangle|\leq O(\epsilon)+\sum_{s=1}^{t_{\epsilon,d}}|\lozenge_{s}|.
\end{equation}
The Lieb-Robinson bound 
\begin{equation}
\left|\lozenge_{s}\right|\leq C'\min(1,e^{\alpha t-r/v})
\end{equation}
can be used to control the remaining term
\begin{equation}
|\langle o_{x}\mid\rho-\rho'\rangle|\leq O(\epsilon)+t_{\epsilon,d}C'\min(1,e^{\alpha t_{\epsilon,d}-r/v})\label{eq:pre_final}
\end{equation}

Here $r$ is the separation between observable $o$ and the perturbation.
We must have $r/v>\alpha t_{\epsilon,d}$ in order to usefully bound
the second term. This inequality gives a bound on the fastest $\epsilon$
can decay with $r$, which in turn is a upper bound on how fast the
RHS of Eq.~\eqref{eq:pre_final} can decay with $r$. In fact, as it
only involves a rescaling of $r$, we can obtain the optimal scaling
if we set $\epsilon$ such that $r/2v=\alpha t_{\epsilon,d}$. $t_{\epsilon,d}$
is monotone in $\epsilon$, so this equation has a unique solution
$\epsilon=\epsilon_{d,r}$. Setting $\epsilon$ to be this value gives
\begin{equation}
|\langle o_{x}\mid\rho-\rho'\rangle|\leq O(\epsilon_{d,r})+\frac{r}{2\alpha v}C'\min(1,e^{-r/2v})\label{eq:master_lppl}
\end{equation}
The second term is exponentially decaying in $r$. The first term
is determined by the details of the $f$ function. We examine some
examples. 

\subsubsection*{Example 1: Faster-than-polynomial decay}

Let $f(t,d)=e^{\alpha d^{a}-\Delta t^{b}}$ where $a,b>0$. Note $t_{d}=(\frac{\alpha}{\Delta})^{1/b}d^{a/b}$.
It is easy to show that 
\begin{equation}
F(t,d)\leq\frac{C''}{b\Delta}e^{\alpha d^{a}-\frac{\Delta}{2}t^{b}}
\end{equation}
for sufficiently large $t$ (``sufficiently large'' will be set
by $b,\Delta$ only). We see that $t_{\epsilon,d}=O(\frac{\alpha d^{a}+\log(1/\epsilon)}{\Delta})^{1/b}$.
On the other hand, solving the equation $r/2v=\alpha t_{\epsilon,d}$
for $\epsilon$ will give $\epsilon_{d,r}=e^{O(\alpha d^{a}-\frac{\Delta}{4\alpha v}r^{b})}$.
On net we find
\begin{equation}
|\langle o_{x}\mid\rho-\rho'\rangle|\leq e^{O(\alpha d^{a}-\frac{\Delta}{4\alpha v}r^{b})}+\frac{r}{2\alpha v}C'e^{-r/2v}
\end{equation}
Thus, if $f$ decays faster than polynomially with $t$, the local
perturbations have an effect that falls off at least as fast as $e^{-O(r^{\min(b,1)})}$.

\subsubsection*{Example 2: Polynomial decay }
Suppose $f(t,d)=Cd^{a}/t^{b}$ where $a,b>0$. $t_{d}=C^{1/b}d^{a/b}$.
Then
\begin{equation}
F(t,d)=\frac{C}{b-1}d^{a}/t^{b-1}
\end{equation}
We require $b>1$ for the boundedness of $F$. This gives $t_{\epsilon,d}=\max(t_{d},(\frac{C}{b-1}\frac{d^{a}}{\epsilon})^{1/(b-1)})$.
We find that $\epsilon_{d,r}=\frac{c'd^{a}}{b-1}/r^{b-1}$ so that
at sufficiently large $r$ we obtain
\begin{equation}
|\langle o_{x}\mid\rho-\rho'\rangle|\leq\frac{c'd^{a}r^{1-b}}{b-1}+\frac{r}{2\alpha v}C'e^{-r/2v},\label{eq:lppl_poly}
\end{equation}
Therefore local perturbations could potentially have an influence
that only decays as a power law $r^{1-b}$ in this case. Nevertheless,
it seems that $b>1$ (i.e., $f(t,d)$ decays faster than $t^{-1})$
is sufficient for a weak form of LPPL.

\subsubsection*{Conjecture: Phase stability requires \texorpdfstring{$f \lesssim t^{-(D+1)}$}.}

Eq.~\eqref{eq:lppl_poly} suggests that $f(t,d)=Cd^{a}/t^{b},b>1$
leads to a form of LPPL with power-law decay $r^{1-b}$. If one makes a series of such local perturbations throughout the system,
and that the effect of such perturbations is additive, that suggests the net effect on observable $o$ would go as $\int_{1}^{\infty}d^{D}r\,r^{b-1}$
where $D$ is the spatial dimension. This diverges if $D>b-1$ in
which case local perturbations can cumulatively act in a divergent
fashion. Therefore, we conjecture that the stability of the phase
requires that $f(t,d)$ decays at least as fast as $t^{-(D+1)}$ at
large $t$. 
\subsection{Long-ranged order}
Finally we repeat our proof that long-ranged order is stable in uniform phases, using the more explicit statement of uniformity Eq.~\eqref{eq:rapid_mix_new}. The proof proceeds as before in Theorem \ref{eq:LRO_theorem} until Eq.~\eqref{eq:LRO3}, which now becomes
\begin{equation}
\bigl|\bigl\langle\tilde{o}_{x}\tilde{o}_{0}|\rho'\bigr\rangle^{\mathrm{c}}-\bigl\langle o_{x}o_{0}|\rho\bigr\rangle^{\mathrm{c}}\bigr|<C f(t,d) 
\end{equation}
where $d$ encodes the support of the correlator with long range order, namely $o_{0}o_{x}$, $C$ is an $O(1)$ constant, and we require $|x|>2vt$. To show the stability of long-range order, it suffices to show that the RHS of the above equation can be made smaller than $D/2$  and for arbitrarily large $x$ by choosing some $t$ independent of $x$\footnote{This last condition comes from ensuring the dressed observables $\tilde{o}_{0,x}$ defined in Theorem.~\ref{eq:LRO_theorem} are local and not growing in size with $x$}. Many functions $f(t,d)$ satisfy these requirements. It suffices to require that $f(t,d)$ decays in time, and to  not increase with the operator diameter ($\sim |x|$).

\section{Example of a topologically non-trivial steady state bundle}\label{app:bundle}

The definition of uniformity is `local'; it places conditions on the relaxation rates of the steady states of \emph{nearby} channels within the steady-state bundle. Can we not just define a phase as being a region of parameter space where the steady states of any pair of channels within the phase rapidly relax to one another? Suppose one adopts this simpler definition. In that case pick a $\mathcal{E}$ in the phase. Any steady state $\rho'$ of any channel in the phase $\mathcal{E}'$ can be unambiguously associated with a steady state $\rho=\rho(\rho')$ which it rapidly relaxes to under the action of $\mathcal{E}'$. We will show that this simpler definition is too restrictive; it appears to exclude a phase we believe exists. 

Let $\mathcal{E}$ be the Toom dynamical map, embedded diagonally
such that $\mathcal{E}(\mid\sigma\rangle\langle\sigma'\neq\sigma\mid)=0$.
Let $\mathfrak{r}_{\theta}$ be a global spin rotation of angle $\theta$
about the Y-axis. Then $\mathcal{E}_{\theta}\equiv\mathfrak{r}_{\theta}\mathcal{E}\mathfrak{r}_{\theta}^{\dagger}$
is an isospectral family of channels, which we suppose to be in
the same phase. $\mathcal{E}_{\theta}$ describes a 1D path through
that (presumably) absolutely stable phase. $\mathcal{E}_{\theta}$
has steady states $P_{\uparrow/\downarrow}^{\theta}\equiv\mathfrak{r}_{\theta}(P_{\uparrow/\downarrow})$,
fully polarised in the $\hat{n}^{\mathrm{T}}=(\sin\theta,0,\pm\cos\theta)$
direction respectively.

We now argue that this putative phase does not conform to the simpler definition of phase described above. If it did, then each of $P_{\uparrow/\downarrow}^{\theta}$ must rapidly relax to states in the span of $P_{\uparrow/\downarrow}^{\theta=0}$ under the standard Toom dynamics. Call these two states $\omega_{\theta,\uparrow/\downarrow}$ respectively. Now $P_{\uparrow/\downarrow}^{\theta}$ have short ranged correlations, therefore $\omega_{\theta,\uparrow/\downarrow}$ must do too as short time evolution cannot introduce long range correlations. The only option consistent with this requirement is for $\omega_{\theta,\uparrow/\downarrow}$ to be some permutation $\eta(\theta)$ of $P_{\uparrow/\downarrow}^{\theta=0}$. As we continuously vary $\theta$ we also continuously vary $\omega_{\theta,\uparrow/\downarrow}$. Hence, the permutation $\eta$ is also continuous in $\theta$ and therefore a constant. 
On the other hand, it is straightforward to see that $\mathcal{E}_{\pi}=\mathcal{E}_{0}$ and $\eta(\pi)=(12)$ because $\mathcal{E}$ happens to be Ising symmetric, but the Ising flip does swap the two steady states. Therefore $\eta$ cannot be a constant, contradicting the fact it is continuous. Therefore, we reach a contradiction; it is not in fact possible for all models in the Toom phase to have their steady states rapidly relax to one another.  

\section{Scaling of steady states in unstable 1D CA}\label{app:soldier}

\begin{figure}
    \centering
    \includegraphics[width=0.4\textwidth]{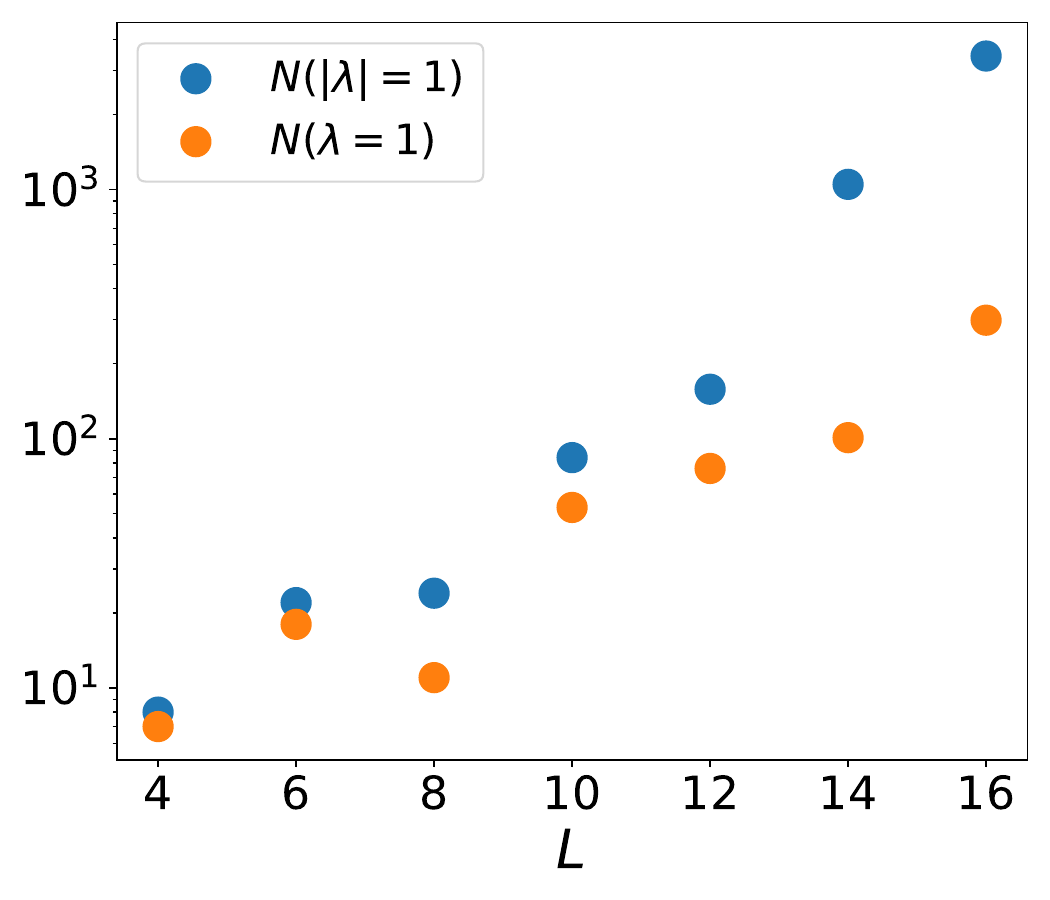}
    \includegraphics[width=0.4\textwidth]{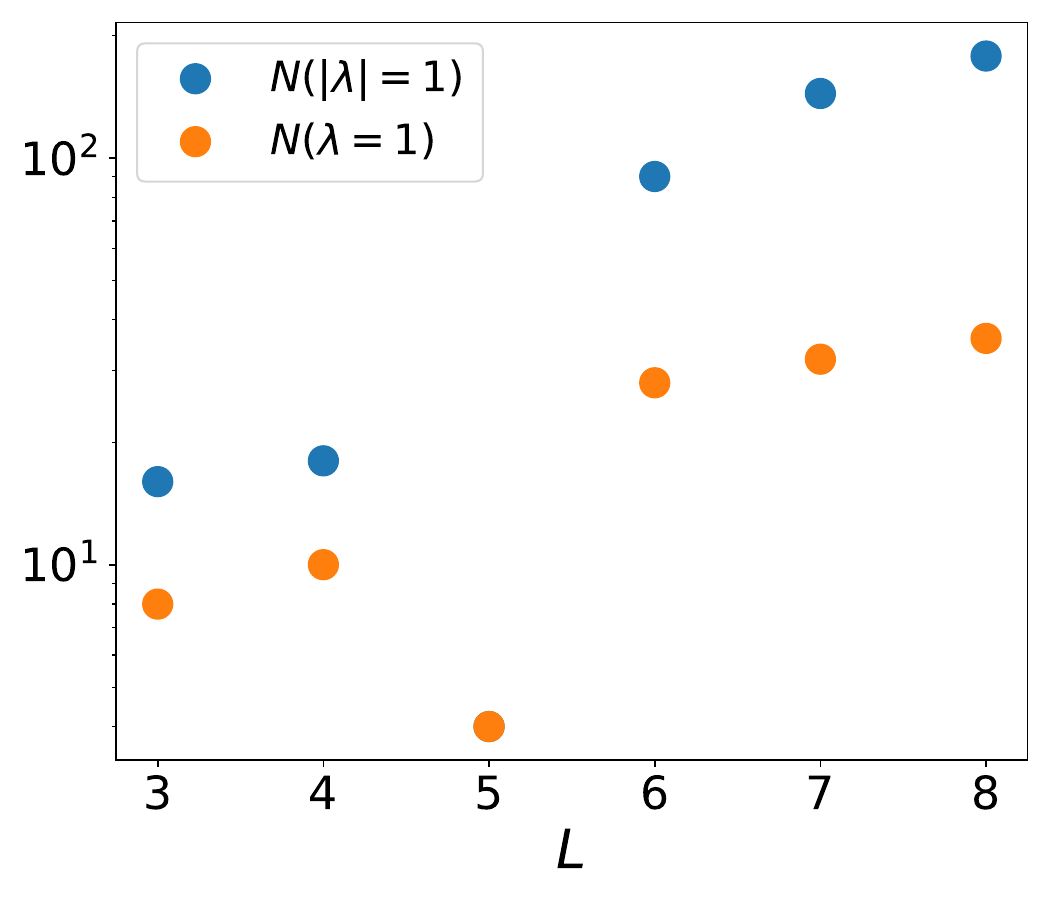}
    \caption{Number of steady states of ``soldier rule'' (left( and ``two-line voting'' (right) models as a function of system size. We separately show the number of steady states with eigenvalue exactly $1$ ($N(\lambda=1)$) and the number of states with eigenvalue modulus one ($N(|\lambda|=1)$) In both cases, we find on-trivial growth as $L$ is increased.}
    \label{fig:soldier}
\end{figure}

We consider two examples of 1D cellular automata which have the property that they are eroding with respect to both the $\ket{\vec{0}}$ and the $\ket{\vec{1}}$ steady states. An island of size $R$ that is an inclusion of the opposite steady state gets eroded in time $O(R)$. Nevertheless, these models are known~\cite{park1997ergodicity} to be unstable in the sense that they have a unique steady state when arbitrarily weak noise is introduced. Thus, if Conjecture~\ref{conj:stability} is to be true, they must exhibit a large number of steady states. 

The two models are named ``soldier rule'' and ``two-line voting''~\cite{park1997ergodicity}. Soldier rule is defined on 1D binary strings $b_x = 0,1$ and is defined by the update rule
\begin{equation*}
    b_x \to \text{Maj}(b_x,b_{x+b_x},b_{x+3b_x}),
\end{equation*}
where $\text{Maj}$ takes the majority value among its three arguments. It is thus a majority of three spins, but which three depends on the value of $b_x$ itself (picture each bit as a soldier facing either left or right and aligning itself with its comrades). Two-line voting is defined on a two-leg ladder (similar to Toom's ladder) with bits $b_{x,j=\pm1}$. It has the update rule
\begin{equation*}
    b_{x,j} \to \text{Maj}(b_{x,-j},b_{x-j,j},b_{x-2j,j}).
\end{equation*}

As said, both models have fast erosion (for both all up and all down steady states) but nevertheless exhibit non-perturbative instability. Thus, by analogy with the Toom's ladder model, we expect them to host a large number of steady states. This is confirmed by our numerics shown in Fig.~\ref{fig:soldier}. We consider both static steady states, with eigenvalue $\lambda=1$ and states with $|\lambda|=1$, both of which grow with $L$ for both models considered. For the soldier rule, where we are able to go up to $L=16$, we see behavior that suggest an exponential growth with $L$; in the two-line voting model we are limited to $L \leq 8$ so there ultimate scaling is less clear.  

\bibliography{main.bbl}
\end{document}